\documentclass[11pt,english]{article}
\usepackage{pst-node}
\usepackage{pst-coil}
\usepackage[english]{babel}
\usepackage{multicol}
\usepackage[margin=0.9in]{geometry}
\usepackage{hyperref}
\usepackage{blindtext}
\hypersetup{
     colorlinks=true,
     linkcolor=blue,
     citecolor=blue,
 }
\usepackage{bm,xfrac,amsmath,amsthm,amssymb,soul, comment, xspace, mathtools}
\allowdisplaybreaks
\usepackage{cleveref}
\usepackage[most]{tcolorbox}
\usepackage{enumitem}
\usepackage{csquotes}
\usepackage[T1]{fontenc} 
\usepackage{footnote,dsfont}
\usepackage{nicefrac,bm}
\usepackage[ruled,vlined,linesnumbered]{algorithm2e}
\usepackage{caption}
\usepackage{thmtools, float}
\usepackage{thm-restate}
\usepackage{graphicx}
\usepackage{pifont,multirow}
\usepackage{textcomp}
\usepackage{lmodern}
\usepackage{bold-extra}
\usepackage{multicol}

\newtcolorbox[auto counter]{invariant}[1][]{
  enhanced,
  breakable,
  fonttitle=\scshape,
  title={Invariant \thetcbcounter},
  colback=white,
  #1
}
\def\confversion{0} 
\usepackage{ifthen}

\newcommand{\ignore}[1]{}

\ifthenelse{\equal{\confversion}{1}}
{
	\newcommand{\conf}[1]{#1}
}
{
	\newcommand{\conf}[1]{\ignore{#1}}
}
\ifthenelse{\equal{\confversion}{0}}
{
	\newcommand{\full}[1]{#1}
}
{
	\newcommand{\full}[1]{\ignore{#1}}
}

\full{
\title{Online Fair Division: Towards Ex-Post Constant MMS Guarantees}

\author{
Pooja Kulkarni
\thanks{University of Illinois at Urbana-Champaign}\\ \texttt{poojark2@illinois.edu} 
\and
Ruta Mehta
\thanks{University of Illinois at Urbana-Champaign.}\\\texttt{rutameht@illinois.edu} 
\and
Parnian Shahkar
\thanks{University of California, Irvine}\\ \texttt{shahkarp@uci.edu}
}
\date{}
}

\conf{
\title{Online Fair Division: Towards Ex-Post Constant MMS Guarantees}

\author{Anonymous Author(s)}
\date{}
}

\newcommand{\I}{{\mathcal I}}

\newcommand{\MMS}{{\sf{MMS}}}

\let\oldnl\nl
\newcommand{\nonl}{\renewcommand{\nl}{\let\nl\oldnl}}

\renewcommand{\nonl}{\renewcommand{\nl}{\let\nl\oldnl}}
\long\def\symbolfootnote[#1]#2{\begingroup%
\def\thefootnote{\fnsymbol{footnote}}\footnote[#1]{#2}\endgroup}

\newcommand{\unsat}{\mathtt{unsaturated}}

\newtheorem{theorem}{Theorem}[section]
\newtheorem{proposition}[theorem]{Proposition}

\newtheorem{lemma}[theorem]{Lemma}

\newtheorem{claim}[theorem]{Claim}

\crefname{algocf}{Algorithm}{Algorithms}
\Crefname{algocf}{Algorithm}{Algorithms}
\theoremstyle{definition}
\newtheorem{definition}[theorem]{Definition}
\newtheorem{example}[theorem]{Example}
\allowdisplaybreaks

\newcommand*{\Ical}{\mathcal{I}}
\newcommand*{\Icalhat}{\widehat{\mathcal{I}}}

\newcommand{\problem}{\textsc{OnlineKTypeFD}}

\newcommand{\remove}[1]{}

\begin{document}
\begin{titlepage}
\maketitle    
\begin{abstract}
We investigate the problem of fairly allocating $m$ indivisible items among $n$ sequentially arriving agents with additive valuations, under the sought-after fairness notion of maximin share (MMS). We first observe a strong impossibility: without appropriate knowledge about the valuation functions of the incoming agents, no online algorithm can ensure \textit{any} non-trivial MMS approximation, even when there are only two agents.

Motivated by this impossibility, we introduce \textsc{OnlineKTypeFD} (online $k$-type fair division), a model that balances theoretical tractability with real-world applicability. In this model, each arriving agent belongs to one of $k$ types, with all agents of a given type sharing the same known valuation function. We do not constrain $k$ to be a constant. Upon arrival, an agent reveals her type, receives an irrevocable allocation, and departs. We study the ex-post MMS guarantees of online algorithms under two arrival models:

\begin{itemize}
    \item \textbf{Adversarial arrivals:} In this model, an adversary determines the type of each arriving agent. We design a $\frac{1}{k}$-MMS competitive algorithm and complement it with a lower bound, ruling out any $\Omega(\frac{1}{\sqrt{k}})$-MMS-competitive algorithm, even for binary valuations.
    \item \textbf{Stochastic arrivals:} In this model, the type of each arriving agent is independently drawn from an underlying, possibly unknown distribution. Unlike the adversarial setting where the dependence on $k$ is unavoidable, we surprisingly show that in the stochastic setting, an asymptotic, arbitrarily close-to-$\frac{1}{2}$-MMS competitive guarantee is achievable under mild distributional assumptions. 
\end{itemize}

Our results extend naturally to a \emph{learning-augmented} framework; when given access to predictions about valuation functions, we show that the competitive ratios of our algorithms degrade gracefully with multiplicative prediction errors.

The main technical challenge is guaranteeing ex-post fairness, i.e., ensuring that \textit{every} arriving agent gets a bundle of sufficient value. 
For this, we design novel approaches built on ideas of maintaining \textit{tentative} overlapping allocations and multi-phase bag-filling, combined with procedures to carefully handle high-valued items, all of which may be of independent interest. 
\end{abstract}

\newpage

\setcounter{tocdepth}{2}
\tableofcontents
\end{titlepage}

\section{Introduction}\label{sec:intro}

Online fair resource allocation is a fundamental problem arising naturally in a variety of real-world scenarios. For instance, during disaster relief efforts—such as those witnessed in the recent Los Angeles wildfires—essential supplies like food, clothing, and medical aid must be distributed promptly as aid requests arrive. Similarly, in cloud computing, a central server must allocate compute and memory resources to a stream of arriving clients with varying demands. In each of these cases, ensuring that allocations are fair in hindsight (i.e., \textit{ex-post}) is critical: e.g. for maintaining trust, satisfaction, and social stability in disaster management, and for ensuring customer retention in industry applications.

Extensive recent work on {online fair division} has primarily focused on online {\em resource (items)} arrival (e.g., \cite{zeng2020fairness, benade2024fair}), and ex-ante guarantees (e.g., \cite{aleksandrov2015online}, \cite{banerjee2022online}) -- see the end of this section for more detailed related works. \cite{walsh2011online}, \cite{kash2014no} and \cite{banerjee2023online} studied online agent arrivals, but for {\em divisible items}. However, these works fail to capture our motivating scenarios where indivisible resources must be allocated to agents arriving online. Indeed \cite{amanatidis2023fair} note that: 
\begin{quote}
\textit{``The alternative model that considers a fixed set of resources and agents who arrive or depart over time has not been considered for indivisible resources partially because it is very challenging to achieve positive results.''} 
\end{quote}

To fill this notable gap, we initiate the study of {\em online discrete fair division with agent-arrival}, where the goal is to allocate a set of $m$ {\em indivisible} items to $n$ sequentially arriving agents. When an agent arrives, she reports her preferences through an additive valuation function. Based on this, she is immediately allocated a bundle of items that she takes and then departs with, rendering the allocation \emph{irrevocable}. The final allocation to all the agents needs to be {\em fair ex-post} under the sought-after fairness notion of {\em maximin share (MMS)}, discussed below.

Maximin share (MMS) \cite{budish2011combinatorial} is a share-based fairness notion: the MMS value of an agent is the highest value she can guarantee herself when she partitions the items into $n$ bundles and receives the bundle with the smallest value. In other words, she chooses a partition that maximizes the value of the least-valued bundle (see \Cref{def:mms}). In an MMS-allocation every agent $i$ receives a bundle worth at least her MMS value. In the offline setting, the non-existence of MMS allocations \cite{procaccia2014fair} motivated extensive work on {\em approximate} MMS allocations \cite{procaccia2014fair, amanatidis2017approximation, garg2019approximating, garg2020improved} where the current best-known guarantee is the existence of $(\frac{3}{4} + \frac{3}{3836})$ multiplicative approximation to MMS \cite{akrami2024breaking}.

In the online setting, we call an algorithm $\alpha$-MMS-competitive (fair) if {\em each} arriving agent receives a bundle worth at least \(\alpha\) times her MMS value\footnote{As soon as an agent arrives and reveals her valuation function, her MMS value and therefore the $\alpha$-MMS value can be computed \cite{woeginger1997polynomial}}.

The {\em holy-grail} would be a constant-MMS-competitive algorithm. 
However, the following simple example demonstrates that no non-trivial approximation of MMS is possible in the absence of additional information (even for the case of two agents with binary valuations):\footnote{Similar observations are well-known in the item-arrival case as well \cite{zhou2023multi}.}

\begin{example}\label{eg:basic}
Suppose $m$ items need to be allocated among {\em two} arriving agents, where $m$ is an even number. The first agent arrives and reports that she likes all $m$ items equally, say at $1$. Her MMS value is computed as $\frac{m}{2}$. To ensure $\alpha$-MMS allocation ex-post for some $\alpha>0$, suppose, we give her a bundle of some $\lceil \alpha \cdot \frac{m}{2} \rceil$ items out of the $m$ items, and she departs with her bundle. 
 
Then, the second agent arrives and reports that she values only two of the \(m\) items at \(1\), so her MMS value is \(1\); however, both of these items could end up in the bundle assigned to the first agent (unless $\alpha\le \frac{2}{m}$). In this case, she derives zero value from any bundle composed of the remaining items, implying only zero-MMS allocation is possible for her. \qed
\end{example}

Given any $\alpha$, we can choose $m > \frac{2}{\alpha}$ and use the above example to show non-existence of $\alpha$-MMS. The above example holds even when we know agent one's complete valuation function and the number of items liked by agent two\footnote{This example also shows that no non-trivial approximation for other important notions like EF1 or Prop1 is possible.}. This example underscores the challenges inherent to an online setting, where agents arrive sequentially and allocation decisions must be made irrevocably. 
In particular, without appropriate information about the valuation functions of the agents that might arrive, we risk losing too much value to the initial agents, precluding any meaningful fairness guarantee. Given the significance of this problem, the following question naturally arises:

\begin{tcolorbox}
\begin{center}
    Q. What information models for online discrete fair allocation with agent arrivals will enable reasonable fairness guarantees, such as constant-MMS approximations, while remaining sufficiently expressive to capture real-world scenarios?
\end{center}
\end{tcolorbox}

\paragraph{Our Contributions.} 
Towards addressing this question, we introduce the $\textsc{OnlineKTypeFD}$ (Online $k$-type fair division) problem. In many applications, including the disaster management and cloud computing discussed above, it is reasonable to assume that the arriving requests are not completely arbitrary and fall into one of a few categories or types\footnote{For instance \cite{jones2017a} groups refugees according to their requirements, and \cite{andersson2020assigning} classifies preferences of localities over refugees by types.}. Accordingly, we 
define $\textsc{OnlineKTypeFD}$ as follows: 
An instance is given by a triple \((N, M, \{v_i\}_{i\in[k]})\), where \(N\) is the indices of agents that will arrive (if \(n\) agents are expected, we assume \(N = \{1, 2, \dots, n\}\)), and \(M\) is the set of indivisible items available upfront to be allocated among the agents as they arrive sequentially. Additionally, we are given \(k\) valuation functions \(v_i: 2^{M} \to \mathbb{R}_+\) for each \(i \in [k]\), with each \(v_i\) corresponding to type \(i\). Every arriving agent belongs to one of these \(k\) types and reveals her type upon arrival. When an agent is of type $i\in[k]$, it means her valuation function is $v_i$. Note that we do not require \(k\) to be a constant.

We study the $\problem$ problem in two different arrival models:
\begin{itemize}[leftmargin=10pt]\phantomsection\label{dfn:advstoch}
    \item \textbf{Adversarial:} A fully-knowledgeable adversary determines the type of each arriving agent.
    \item \textbf{Stochastic:} The type of each arriving agent is drawn independently from an underlying, possibly unknown, distribution $D = (p_1, \ldots, p_k)$ over the \(k\) types. 
\end{itemize}

For the adversarial arrival model, we design a $\frac{1}{k}$-MMS-competitive algorithm, and complement it with a lower bound of $\frac{2}{\sqrt{k}-2}$. The lower-bound holds even for {\em binary} valuations, i.e., the value of an item is either $0$ or $1$ for every agent. In the offline setting, with binary valuations, even when the valuations are submodular, polynomial time algorithms that achieve MMS are known. This contrast highlights the complexity introduced by the online nature of our problem.

Despite this, for the stochastic arrival model, given $D$, we design a constant-MMS competitive algorithm, where we prove the constant can be arbitrarily close to ~$\frac{1}{2}$ under mild assumptions on the minimum probability in $D$. In particular, the algorithm guarantees {\em ex-post} $\sim$$\frac{1}{2}$-MMS allocation to {\em all} the agents with high probability. Here we manage to bound the failure probability by an inverse-exponential, {\em i.e.,} $(\frac{1}{e^{n^c}})$ for some constant $c>0$! We next extend the algorithm and analysis to the case where the underlying distribution $D$ is \textit{unknown}, while still achieving a competitive ratio arbitrarily close to \(\frac{1}{2}\). To the best of our knowledge, these are the first constant-competitive ex-post fairness guarantees for an online fair allocation problem with indivisible items. We note that, barring the MMS value computation for each type through a PTAS \cite{woeginger1997polynomial}, all our algorithms run in polynomial time.

Further, our results extend naturally to a \emph{learning-augmented} framework, where the \(k\) valuation functions are treated as predictions rather than as known inputs. We demonstrate that in this setting, 
our competitive ratios degrade gracefully with multiplicative prediction errors.

Technically the problem is complex given the uncertainty of future demands in addition to the combinatorial complexities well-known even in the offline setting (see survey \cite{amanatidis2023fair}). Therefore, we devote Section \ref{sec:tech-overview} to give an overview of the algorithmic approaches and analysis techniques we develop to work through the nuanced complexities, before proving the formal results in Sections \ref{sec:adversary}, \ref{sec:knownD}, and \ref{sec:unknowndist}, for the adversarial, stochastic with known-distribution, and unknown-distribution settings respectively. Section \ref{sec:prelims} sets up notations, the model, and useful properties.

\paragraph{Other Related Work.}\label{sec:relwork}
In this section, we discuss some of the related works. This is not meant to be an exhaustive survey; we highlight some works that are most related to ours. For further references, we also refer the readers to surveys \cite{aleksandrov2020online, amanatidis2023fair}.

\medskip

\noindent \textit{Online fair division with agent arrivals.} The existing online fair division literature with agents arriving online has focused on divisible items: \cite{walsh2011online} studied fair division of a single divisible item to $n$  online arriving agents. \cite{kash2014no} studied dynamic fair division where they have a set of divisible items available offline, and the agents arrive but do not depart from the system. \cite{sinclair2022sequential} and \cite{banerjee2023online} also classify agents into types according to their valuation functions, similar to our approach. \cite{friedman2015dynamic, friedman2017controlled} consider achieving fairness while minimizing disruptions where a disruption is reallocation of a previously allocated resource. In a related work, \cite{bogomolnaia2019simple} studied the case where agents' utility profiles are drawn randomly from a distribution. \cite{donahue2020fairness} in a different theme consider the closely related setting where uncertain number of agents from different ``groups'' arrive. They define fairness as individuals from different groups getting resource with equal probability.

\medskip
    
\noindent \textit{Online fair division with indivisible items.} While most of the current literature in online fair division focuses on divisible items, there are works with indivisible items. These mostly focus on the item-arrival setting. Among these, \cite{benade2018make} tries to minimize total envy over a period of time, \cite{he2019achieving} studies minimizing re-allocations to achieve EF1, \cite{zhou2023multi} studies MMS in online setting for indivisible items and chores, and \cite{zeng2020fairness} studies fairness and efficiency trade-offs over time with values drawn from a random distribution. \cite{procaccia2024honor} and \cite{yamada2024learning} explore online fair division with indivisible online items through the lens of bandit learning. 

\medskip

\noindent Other online resource allocation models that we don't discuss here are used for allocating divisible items that arrive online among fixed set of agents. To the best of our knowledge, there are no works on settings with agents arriving online and indivisible items.
\section{Notation and Preliminaries}\label{sec:prelims}

For any positive integer $n$, let $[n] = \{1, 2, \ldots, n\}$, and for two positive integers $i,j$ where $i<j$, let $[i,j] = \{i, i+1, \cdots,j\}$. For a set or list \( N \), let $|N|$ denote the number of elements in $N$.

\subsection{Problem Setting: Online Fair Division} 

\begin{definition}[\textsc{OnlineKTypeFD}]\label{dfn:onlinektype}
An instance of the \textsc{OnlineKTypeFD} problem is defined as
\(
\I = ([n], M, \{v_i\}_{i\in[k]}),
\)
where:
\begin{enumerate}
    \item \([n] = \{1, 2, \dots, n\}\) is the indices of agents who arrive sequentially.
    \item \(M\) is the set of indivisible items available upfront.
    \item \(\{v_i\}_{i\in[k]}\) is a collection of valuation functions of $k$ types, where \(v_i:2^M \to \mathbb{R}_{\ge 0}\). Each arriving agent belongs to one of these \(k\) types and reveals her type upon arrival; if an agent is of type \(i\), her valuation function is \(v_i\).
\end{enumerate}
In this online model, each arriving agent is immediately and irrevocably allocated a subset of the items in \(M\), after which she departs permanently. Our goal is to ensure that all agents receive a \emph{fair} share of the items. We study two variations of the \textsc{OnlineKTypeFD} problem: the \emph{adversarial model} and the \emph{stochastic model}.
\begin{itemize}[leftmargin=*]
    \item \textbf{Adversarial:} A fully-knowledgeable adversary determines the type of each arriving agent.
    \item \textbf{Stochastic:} The type of each arriving agent is drawn independently from an underlying, possibly unknown, distribution $D = (p_1, \ldots, p_k)$, where she is of type $i$ with probability $p_i$. 
    Without loss of generality, we assume 
    \(
    p_1 \geq p_2 \geq \cdots \geq p_k,
    \)
which implies that \(p_1 \geq \frac{1}{k}\) and \(p_k \leq \frac{1}{k}\).
\end{itemize}
\end{definition}

For type $i\in[k]$ agents, \(v_i\) expresses their preferences over subsets of \(M\). We assume that these valuations are \emph{additive}: for any \(S \subseteq M\),
$v_i(S) = \sum_{g \in S} v_i(\{g\}).$
For ease of notation, we write \(v_i(g)\) instead of \(v_i(\{g\})\) for every item \(g \in M\).

\noindent

\begin{definition}[Arrival Sequences]
In an \textsc{OnlineKTypeFD} problem, the \emph{arrival sequence} indicates the types of agents and the order in which they arrive, which is unknown to the algorithm beforehand. 
\end{definition}

\subsection{Fairness Notion: Maximin Share (MMS)}\label{sec:prel-MMS}

For any set \(S\) of items and a positive integer \(d\), let \(\Pi_d(S)\) denote the collection of all partitions of \(S\) into \(d\) bundles. 
\begin{definition}[Maximin Share (MMS)]\label{def:mms}
Given an instance $\I = ([n], M, \{v_i\}_{i\in[k]})$ of \textsc{OnlineKTypeFD} problem, the \emph{maximin share} (MMS) value for a valuation function \(v_i\) is defined as
\[
\MMS_{v_i}^{n} = \max_{P \in \Pi_n(M)} \min_{j \in [n]} v_i(P_j).
\]
\end{definition}
As all type $i$ agents have the same valuation function $v_i$, the MMS value of these agents is also defined as $\MMS_{v_i}^{n}$. When the instance $\I$ is clear from the context, we use the notation $\MMS_{v_i}^n$ as $\MMS_i(\mathcal{I})$ or $\MMS_i$. For type $i$ (or a type $i$ agent), an MMS partition $P^i = (P^i_1, P^i_2, \ldots, P^i_n)$ satisfies $\MMS_i = \min_{j \in [n]} v_i(P^i_j)$. 

Since the set of valuation functions \( \{v_i\}_{i\in[k]}\) is provided as part of the input for an \textsc{OnlineKTypeFD} problem, the (sufficiently close approximate) MMS value for each valuation function can be computed in advance using the polynomial time approximation scheme (PTAS) described in \cite{woeginger1997polynomial}, prior to the arrival of any agents. To simplify the exposition, we assume that the MMS value for every type is known in advance. This is in fact the $(1-\epsilon)$-approximate MMS value and all our results go through with this additional factor.

\begin{definition}[\(\alpha\)-MMS competitive algorithm]
In an \textsc{OnlineKTypeFD} problem, for \(0 < \alpha \leq 1\), an algorithm is \emph{\(\alpha\)-MMS competitive} if it returns an allocation where every agent receives a bundle that she values at least $\alpha$ times her MMS value. 
\end{definition}

\begin{definition}[Normalized Instance]
An input instance $\I = ([n], M, \{v_i\}_{i\in[k]})$ of an \textsc{OnlineKTypeFD} problem is \emph{normalized} if for every $i\in[k]$, \(\MMS_i = 1\) and the total value $v_i(M)=n$  for every type \(i\in[k]\). The former also implies that for every $g\in M$ and any valuation function $v_i$, $v_i(g)\leq 1$.
\end{definition}

\noindent{\bf Normalization is without loss of generality (wlog).} Given a non-normalized input instance $\I = ([n], M, \{v_i\}_{i\in[k]})$, for any valuation function $v_i$, the corresponding normalized valuation $v'_i$ is constructed by computing an MMS partition \(P^i = (P^i_1, \ldots, P^i_{n})\) for \(v_i\) and then rescaling the valuations so that for every \(j \in [n]\) and for every item \(g \in P^i_j\), $v'_i(g) = \frac{v_i(g)}{v_i(P^i_j)}$. Thereby, the corresponding normalized input instance is defined as $\I' = ([n], M, \{v'_i\}_{i\in[k]})$. Lemma 4 in \cite{simple} shows that for any bundle $b\subseteq M$, $v_i(b) \ge v'_i(b) \cdot \MMS^{n}_{v_i}$. By construction, the MMS value of every type in the normalized instance is $1$, i.e. for all $i\in[k]$, $\MMS^{n}_{v'_i} =1$. Therefore, if $v_i'(b)\geq \alpha \MMS^{n}_{v'_i}$, we have that $v_i(b) \ge \alpha \cdot \MMS^{n}_{v_i}$. This implies that an \(\alpha\)-MMS competitive algorithm applied to a normalized instance guarantees that each agent receives a bundle valued at least \(\alpha\) times her MMS value in the original (non-normalized) instance. Therefore, we assume throughout this paper that the inputs for the \textsc{OnlineKTypeFD} problem are normalized instances. 

\begin{definition}[Claiming a Bundle]
In an \textsc{OnlineKTypeFD} problem where the goal is to obtain an \(\alpha\)-MMS guarantee given a normalized input instance, a type $i$ or a type $i$ agent is said to \emph{claim} a bundle \(B\) if $v_i(B) \ge \alpha$.
\end{definition}

\begin{definition}[High-valued Items]
In an \textsc{OnlineKTypeFD} problem where the goal is to obtain an \(\alpha\)-MMS guarantee given a normalized input instance, an item $g$ is considered \emph{high-valued} by type $i$ if $v_i(g) \ge \alpha$. Items that are high-valued by all $k$ types are \emph{universally high-valued} items. 
\end{definition}  

\begin{definition}[Bag-filling]
Given a pool of items \(R\), a bag-filling procedure is a process in which a new, empty bag is created and items from \(R\) are added arbitrarily until a predetermined stopping condition is met.
\end{definition}  
Now let us provide a useful concentration bound, used in our analysis later.
\begin{lemma}[Corollary 3.4 of \cite{vondrak2010noteconcentrationsubmodularfunctions}]
\label{cor:nonmonotone-tails}
If $Z = f(X_1,\ldots,X_n)$ where $X_i \in \{0,1\}$ are independently random and $f$ is non-negative
submodular with marginal values in $[-1,1]$, then for any $\delta>0$, $\Pr[Z \leq (1-\delta) \mathbb{E}[Z]] \leq e^{-\delta^2 \mathbb{E}[Z] / 4}.$
\end{lemma}
\section{Technical Overview}\label{sec:tech-overview}
Our key technical contributions are for the stochastic model. However, we first discuss the problem in adversarial model, for a two-fold reason: (1) To bring out the nuanced complexity of the MMS problem in the online setting, and acclimatize the reader with the model without getting caught up in the probabilistic arguments of the stochastic model. (2) Our stochastic model with unknown distributions (\Cref{sec:unknowndist}) uses the adversarial algorithm as a subroutine. 

We note that given a {\em normalized instance} all our algorithms run in polynomial time. As discussed in Section \ref{sec:prel-MMS} any given instance can be normalized wlog, however it requires the computation of MMS values for each type $i\in [k]$. The latter can be done using a PTAS \cite{woeginger1997polynomial}. Thus, in general, our algorithms are polynomial-time barring the MMS value computation via a PTAS. 
\medskip

\noindent{\bf High-level challenges in contrast to offline setting.} The concepts of {\em ordered instances} \cite{garg2019approximating}, {\em single item reductions} and {\em bag-filling} are central to offline MMS algorithms. First, it is no more without loss of generality to assume ordered instance, because the allocations are irrevocable. As a consequence, reductions beyond single-item are no longer applicable. Furthermore, the bag-filling, if used, has to be much more carefully done due to the uncertain future agent types: Say we are filling a bag for an arrived agent of type $i$. If it becomes valuable for another type $j$, it is unclear whether to save it for type $j$ agent who may or may not arrive in the future.
\medskip

\subsection{{Adversarial Model}} 
In this model an all-knowing adversary decides the type of each arriving agent. Here, our first result is a $\frac{1}{k}$-approximation, which is a constant-approximation when $k$ is a constant. 
\noindent\begin{restatable}{theorem}{thmadvpos}\label{thm:adv}
    Given an instance, $\I=([n], M, \{v_i\}_{i \in [k]})$ of the \text{$\problem$} problem in the adversarial model, Algorithm \ref{alg:adversarial-alg} is $\frac{1}{k}$-MMS-competitive.
\end{restatable}

\noindent \textbf{Algorithm and Analysis Idea.} As discussed in Section \ref{sec:prel-MMS}, we can assume wlog that instance $I$ is normalized. and hence MMS$_i=1$ and $v_i(M)=n$ for each type $i\in [k]$. Since the arrival sequence can be arbitrary, one approach is to create $nk$ agents, $n$ of each type. However, MMS is very fragile with respect to (wrt) the number of (\#) agents, {\em e.g.,} if an agent values exactly $n$ items at $1$ each, then adding even one extra agent drops her MMS value from one to zero. This renders any approach appealing to duplicating agents tricky to analyze. Instead, we appeal to {\em tentative allocations} by carefully extending the concept of {\em bag-filling}, so as to ensure that at any point in time, there is enough value for any possible future arrival sequence of agents.

\medskip

\noindent \textit{Tentative Allocations.} Throughout the algorithm, we maintain a set \( G_i \) of tentative allocations for each type \( i \in [k] \). Each \( G_i \) consists of disjoint bundles such that for any \( B \in G_i \), \( v_i(B) \geq \frac{1}{k} \), though a bundle $B$ may belong to multiple $G_i$s.
\medskip

\noindent Ideally, if we could construct tentative allocations where \( |G_i| \ge n \) for all \( i \in [k] \), ensuring that every \( B \in G_i \) satisfies \( v_i(B) \geq \frac{1}{k} \), a \(\frac{1}{k}\)-competitive guarantee would follow: each arriving agent of type \( i \) would receive a bundle from \( G_i \), and updates would ensure that future agents still have access to feasible allocations -- whenever a bundle is allocated, it would be removed from all tentative sets containing it. This reduces the size of each $G_i$ by at most one ensuring the invariant that $|G_i| \ge$ (\# remaining agents) for all $i\in[k]$.

However, this fixed structure is infeasible because even overlapping bundles across types can lead to significant value loss. The following example demonstrates that pre-saturating the types fails to guarantee an approximation of $\frac{1}{k}$ in an adversarial setting.

\begin{example}\label{eg:adv-counter}
    Consider an instance with $k=2$, $3n$ items. Type $1$ values $2n$ of these items at $\frac{1}{2}-\epsilon$ each and the rest $n$ at $2\epsilon$ each. Therefore, the MMS value is $1$ and the MMS allocation combines in each bag $2$ items of $\frac{1}{2}-\epsilon$ value and $1$ item of value $\epsilon$. Now suppose, $\epsilon < \frac{1}{4n}$ for this type and the second type bundles all of these $2\epsilon$ valued items into a single bag. This bundle will not belong to type $1$'s tentative allocation since $n \cdot \frac{1}{2n} < \frac{1}{2}$. Further, suppose the second type bundles any of the remaining $n-1$ items, each of value $(\frac{1}{2}-\epsilon)$ for type $1$ into $n-1$ bags. These bundles too will not be included in $G_1$. Now, the value left for creating bundles for type $1$ is: $n+1$ items each of $(\frac{1}{2} - \epsilon)$ value and we cannot create $n$ bundles of value $\frac{1}{2}$.
\end{example} 

To overcome this challenge, we construct tentative allocations {\em dynamically}, adapting them on-the-fly. Towards this, we define the notion of \emph{(over) saturated types}.  

\medskip

\noindent \textit{Saturated Types.} A type is considered saturated (or over saturated) if the number of bundles in its tentative set equals (or exceeds) the number of agents left to arrive.  

\medskip  

\noindent When an agent arrives with type $i$, if there exists $B\in G_i$ then $B$ is allocated to the agent. Otherwise, a new bundle is formed from unassigned items through {\em bag-filling}. In this process whenever the bundle is valuable at more than $\frac{1}{k}$ to an unsaturated type(s), but not the current agent, it is tentatively allocated to those type(s). At the end of every iteration of the algorithm, we release excess bundles from any over saturated type. We show that these steps together maintain the invariant: for each type $i\in [k]$, either it is saturated or  $v_i(\mbox{remaining items}) \ge$ (\# remaining agents). This ensures a bundle worth $\frac{1}{k}$ to every agent (see \Cref{sec:adversary} for details).

\medskip

\noindent \textbf{Lower bound.}  
Given prior knowledge of valuation types of agents and the tractability of MMS in the offline setting, one would hope to achieve a constant independent of $k$ in the approximation factor. However, we show a lower bound of \( O(\frac{1}{\sqrt{k}}) \) instead, implying that our algorithm is achieving nearly the best possible guarantee. 
\begin{restatable}{theorem}{thmadvneg}\label{thm:adv-negative}
    Given an instance, $\I=([n], M, \{v_i\}_{i \in [k]})$ of the $\problem$ problem in the adversarial model no online algorithm is better than ${\frac{{2}}{\sqrt{k} - 2}}$-MMS-competitive.
\end{restatable}
\noindent\noindent Notably, our lower bound holds even for binary valuations and a constant number of types. The key idea is to construct an instance with binary valuations, where valuations are not normalized. In this instance, we have one type with a high-MMS value, where this will be the first agent's type, and multiple other types with an MMS value of 1. Since valuations are binary, the high-MMS type values many items, while every other types value exactly \( n \) items. By appropriately selecting these \( n\) items for the remaining types, we ensure that any allocation better than  $\frac{2}{\sqrt{k}-2}$-MMS for the first agent will result in one of the other types losing two items. The adversary then sends all remaining agents of this type, forcing the last agent to receive zero value (see \Cref{subsec:adv-lower-bound} for details).

\subsection{Stochastic Arrivals with Known Distribution}
Looking for better guarantees for scenarios arising in practice, we consider a stochastic setting where along with an $\problem$ instance, we are given a distribution $D = (p_1, \ldots, p_k)$ such that every arriving agent is of type $i$ with probability $p_i$. In particular, we prove the following.

\begin{restatable}{theorem} {thmknownD}\label{thm:knownD}
Given an instance, $\I=([n], M, \{v_i\}_{i \in [k]})$ of the $\problem$ problem in the stochastic arrival model with known distribution \(D\), for any constant \(0 <\epsilon < \frac{1}{2}\) such that 
\(
p_k = \omega\Bigl(\frac{n^\epsilon k}{\sqrt{n}}\Bigr),
\)
then, for any \(\eta > 0\) there exists an \(n(\eta)\) such that for all \(n > n(\eta)\), \Cref{alg:knownD} is \(\frac{1}{2(1+\eta)}\)-MMS competitive with probability at least \(1 - o\Bigl(\frac{1}{e^{n^{1.5\epsilon}}}\Bigr)\). 
\end{restatable}

\noindent Although the above theorem seems to suggest that $1/2$ factor is achieved for asymptotically large $n$, for a slightly weaker factor the bound on $n$ is quite reasonable. For example, 
when $p_k \geq \frac{k}{n^{0.2}}$ the guarantee of the algorithm is $\frac{1}{3}$-MMS competitive for any $n \geq 59$, and $\frac{1}{2.1}$-MMS competitive for any $n\geq 7195$.\footnote{This is obtained by choosing $\epsilon= 0.001$, $\eta = 0.5$ for the former guarantee and $\eta = 0.1$ for the latter. The guarantees provided hold with probability at least $1-o\Bigl(\frac{1}{e^{n^{0.0015}}}\Bigr)$} Furthermore, the lower bound on $p_k$ implies an upper bound on $k$, however it is sublinear, i.e., $k= o(n^{0.25})$. We also note that our probability of failure is inverse-exponential!

\medskip

\noindent \textbf{Algorithm and Analysis Idea.} Our goal here is to use the distribution information to achieve better guarantees with (very) high probability (whp). The tricky part is to handle the events when relatively few, i.e., $o(\log n)$, agents are left to arrive. In these cases we cannot infer any statistical guarantees about the arrival sequence. One easy way to handle this would be to ensure fairness to all but $o(n)$ agents, or ensure ex-ante fairness. But our goal is {\em ex-post} fairness to {\em all} agents! We therefore take a multi-phase approach in \Cref{alg:knownD}. 

First, observe that given the distribution $D$, the expected number of agents of type $i$ is $np_i$ for any $i \in [k]$. Further, using Chernoff bound we get that \emph{for all $i \in [k]$}, with high probability at most $M_i = \lfloor np_i + n^\epsilon \sqrt{np_i}\rfloor$ agents will be of type $i$. 

We build upon the ideas of {\em tentative allocations} and {\em saturated types} from the previous section.
If we directly create {\em overlapping} tentative allocations \( G_i \) for each type $i$, like in the previous section, then it is unclear when to declare a type saturated. Using \( M_i \) as a threshold is infeasible because $G_i$ may loose shared bundles to other types leading to faster than tolerable depletion. The alternative of setting saturation at \( n \) ignores the underlying distribution. To address this, we can try to create \emph{disjoint} $G_i$s, and declare a type $i$ saturated when enough bundles are reserved for agents of that type, i.e. $|G_i| = M_i$. However, then the challenge is to bound the losses. 

Creating $M_i$ {\em exclusive} bundles for each type $i$ requires us to create in total $\sum_{i \in [k]}M_i \sim \sum_{i \in [k]}np_i + \sum_{i \in [k]} n^\epsilon \sqrt{np_i} = n + \sum_{i \in [k]} n^\epsilon \sqrt{np_i} > n$ many bundles of some constant approximate MMS value. This clearly is impossible for a binary instance where all types like exactly $n$ items, and $n- O(1)$ items are valued by all types. 

We carefully combine the two approaches of {\em overlapping} and {\em exclusive} bundles, handling high-valued and low-valued items through a multi-phase algorithm.

\begin{itemize}[leftmargin=*]
     \item \textbf{Phase 1: Universally high-valued items.} In this phase, items valued highly by all types are allocated in an online manner appealing to {\em overlapping} $G_i$s.
     \item \textbf{Phase 2: High-valued items.} In this phase, items valued highly by some, but not all types are allocated {\em exclusively} to a $G_i$, for some type $i$ that is not saturated, while ensuring that the losses can be bounded later. 
 \item \textbf{Phase 3: Low-valued items.} After Phase 2, all the items are low-valued for the unsaturated types. These are {\em exclusively} bundled via bag-filling.
 \end{itemize}
 We now give further nuances of each phase, and how they work together. 
\medskip

\noindent \textit{Phase 1: Universally high-valued items.} If there are $n'$ universally high-valued items, then it is natural to allocate them as singletons to the first $n'$ agents. Indeed we do exactly this in Phase I, unless $n'$ is too large, but less than $n$. The difficulty with large $n'$ is as follows: if $n-n'$ is small, say a constant, then every sequence of remaining $n-n'$ agents has a constant probability of occurring. This hints at a reduction to the adversarial case where no better than $\Omega(\frac{1}{\sqrt{k}})$ is possible. 

To circumvent this, whenever $n'$ is {\em large}, but less than $n$, we compute an $\alpha$-MMS partition for type $1$, the type whose arrival probability is the largest. In this partition, at least $n - n'$ bags exclude the universally high-valued items. We keep these $n-n'$ bags for type $1$ agents. Since $n p_1 \geq n/k >> n-n'$, there will be enough agents of type $1$ whp to consume these. Now, it is safe to allocate $n'$ universally high-valued items as before to agents of other types and the remaining agents of type $1$. Therefore, if $n'$ is large, the algorithm terminates here. Otherwise we proceed with Phase 2. 

\medskip

After phase 1, phases 2 and 3 essentially create exclusive bundles of high and low-valued items, respectively. To bound overall loss and prove that sufficient number of bundles can be created, we handle the high-valued and the low-valued items separately as follows.

\medskip

\noindent \textit{Phase 2: High-valued items.} To ensure that no type loses too much value in this phase, we create a careful ordering of types according to which they select their high-valued items -- in order, each type picks their high-valued items as singleton bundles until saturated, {\em} i.e., $|G_i|=M_i$. If a type remains unsaturated by the end of this phase, each remained item is valued less than $\alpha$-MMS ($\alpha$ times its MMS value) for it. 

The ordering is done in one of two ways: (1) If there are two types such that the set of their high-valued items have small enough overlap, we place them at the two ends (2) If there are no two types of this kind, we keep the agent with smallest arrival probability as last. In both cases, bounding the loss of the last type is the trickiest. In the former case, losses to the other types are bounded using a combination of the minimum probability bound and low overlap with the first type. In the latter case, through a careful counting argument, we demonstrate that if the last type remains unsaturated by the end of phase 2, the number of high-valued items available for this type must be small. This follows from the absence of universally high-valued items, and large overlaps between the set of high-valued items for any pairs of types.

\medskip

\noindent \textit{Phase 3: Low-valued items.} From a combinatorial point of view, the low-valued items are considerably easier to handle. We construct bundles by bag-filling from remained items until an unsaturated type claims a bundle, meaning it values the bundle at least $\alpha$-MMS. Since each unsaturated type now values each item at at most $\alpha$-MMS, we can bound the loss in every bag by $2\alpha$-MMS using standard arguments. Since $\alpha = \frac{1}{2(1+\eta)}$ this loss is at most $1$. Finally, to bound the overall loss despite creating more than $n$ bundles, we crucially use the bound on the minimum probability and the dependence on $\eta$.\\

While we have emphasized the combinatorial aspects of the algorithm and analysis, the full analysis requires a nuanced accounting of loss and intricate probabilistic arguments. We defer these details to the main technical section (\Cref{sec:knownD}).

\subsection{Stochastic Arrivals with Unknown Distribution}
Building on the case of known distributions, we next handle the case when distribution $D$ is unknown. For this, we essentially demonstrate that MMS retains the flexibility to first \emph{learn} the distribution while achieving the \emph{same} competitive guarantee as before. However, this comes at the cost of a more stringent constraint on the minimum probability bound, which, in turn, imposes a tighter limitation on the number of types. Despite these limitations, our upper bound on $k$, though smaller than in the known distribution case, remains sublinear. Formally, we prove the following theorem:
\begin{restatable}{theorem}{thmunknownD}\label{thm:unknownD}
    Given an instance, $\I=([n], M, \{v_i\}_{i \in [k]})$ of the $\problem$ problem in the stochastic arrival model with an unknown distribution \(D\), for any constant \(0<c<0.1\) such that 
    \(
    p_k = \omega\Bigl(\frac{k}{n^{\frac{2}{9}(1-c)}}\Bigr),
    \)
     for any \(\eta>0\), there exists an \(n(\eta)\) such that for all \(n>n(\eta)\), \cref{alg:unknownD} is \(\frac{1}{2(1+\eta)}\)-MMS competitive with probability at least \(1 - o\Bigl(\frac{1}{e^{n^{c/2}}}\Bigr)\).
\end{restatable}
 
\noindent \textbf{Algorithm and analysis idea.} Our algorithm operates in two main stages:  

\noindent (1) \textit{Learning Stage:} We first observe a sufficient number of agents to estimate the type distribution. Meanwhile, we allocate bundles valued at least $\alpha$-MMS to all agents arriving at this stage.

\noindent (2) \textit{Allocation Stage:} Using the learned distribution, we apply the algorithm designed for the known distribution setting to obtain an $\alpha$-MMS allocation for the remaining agents.  

A key challenge is ensuring that no type loses excessive value during the learning stage. For instance, if agents of type 1 derive value only from $n$ specific items, we must avoid depleting all of them upfront. Further, this should hold for all types. To address this, we will use randomization. 

Recall that value of any item for a type is at most one. Then, if each item is randomly assigned to two baskets with probability $p$ and $(1 - p)$, then, with high probability, the first basket contains an almost-$p$-fraction of the type’s total value, while the second holds an almost-$(1 - p)$-fraction\footnote{This will hold up to submodular valuations.}. This follows because when viewing MMS as a function of the items in the instance, this function satisfies a Lipschitz property. This allows us to use a Chernoff-type concentration bound to bound the value of a single type in a random bag. Since Chernoff gives exponentially strong bounds on probability, we can use a union bound over all types to ensure value is preserved for all types using the \emph{same} random bag. Leveraging this, we divide the items into two baskets, say $B_1$ and $B_2$. Our high-level plan is to use $B_1$ in the  {\em learning stage}, and $B_2$ in the {\em allocation stage}. 

The probabilities with which the items are distributed between $B_1$ and $B_2$ should be carefully chosen: If $n'$ agents are handled in the learning stage, then $B_1$ should have sufficiently {\em large} value to be distributed to them. This is because the learning-stage is essentially like the adversarial setting, where we have to run the \Cref{alg:adversarial-alg} on instance with $n'$ agents and $B_1$ as the set of items. To ensure that the corresponding allocation to these $n'$ agents is constant-MMS in the original instance, $B_1$ should have value at least $\Theta(n'\cdot k)$ for all the types whp. This follows, if we assign each item to $B_1$ with probability $\frac{2kn'}{n}$. To ensure that this is not too large a loss for the agents arriving in allocation stage while still enough to learn the distribution $D$, we choose $n' \sim n^{\frac{2}{3}}$.

\noindent {\em Handling high-valued items.}
In executing our high-level plan, high-valued items again create issues as in the previous settings. 
If there are a large number of universally high-valued items, then a large fraction of $B_1$ would be formed by these. In particular, $B_1$ may contain far more than $n'$ such items. Clearly keeping these exclusively for the first $n'$ agents will leave insufficient items for the remaining agents. We therefore introduce a stage that isolates universally high-valued items and learns from them for as long as possible. If there are at least $n'$ such items, we use them for learning the distribution. Otherwise, we allocate these items to the earliest arriving agents, remove them from consideration, and then randomly distribute the remaining items into two baskets as discussed earlier. Then, by observing the types of the next $n'$ agents, we estimate the distribution. Meanwhile, we allocate items from $B_1$ to these agents.

In the second (allocation) stage, we are in the case of known-distribution and we apply the algorithm from previous section (\Cref{alg:knownD}) to obtain an $\alpha$-MMS allocation.

While the complete analysis is nuanced, we hope this overview provides an initial understanding for the reader, and we refer them to \Cref{sec:unknowndist} for the full details.

\subsection{Extension to Learning-Augmented Framework}
Now by proving the following lemma, we will show that all of our results extend to a learning-augmented framework. 

\begin{restatable}{lemma}{learningaugmented}\label{lem:learningaugmented}
Consider two inputs of an \textsc{OnlineKTypeFD} problem, \(\I^{e} = ([n], M, \{v^{e}_i\}_{i\in[k]})\) and \(\I = ([n], M, \{v_i\}_{i\in[k]})\), and suppose there exists a constant \(\beta > 1\) such that for every \(i\in[k]\) and every \(g\in M\), \(\frac{1}{\beta}v_i(g) \le v^{e}_i(g) \le \beta v_i(g)\). Then, if an algorithm is \(\alpha\)-MMS competitive on \(\I^{e}\), it allocates each agent a bundle valued at least \(\frac{\alpha}{\beta^2}\) times her MMS value with respect to the instance \(\I\).
\end{restatable} 
\begin{proof}
    Consider any type $i\in[k]$ agent, and a bundle $b$ she receives in the online algorithm. Let $\mu^e_i = \MMS_{i}(\I^{e})$, and $\mu_i = \MMS_{i}(\I)$. Let us consider the MMS partition of type $i$ with respect to $\I$, $P^i = (P^i_1, P^i_2, \ldots, P^i_n)$ satisfying $\mu_i = \min_{j \in [n]} v_i(P^i_j)$. Hence, for any $j\in [n]$, $v_i(P^i_j) \geq \mu_i$, and since valuations are linear, $\beta v^e_i(P^i_j) \geq v_i(P^i_j)$. These imply $v^e_i(P^i_j) \geq \frac{\mu_i}{\beta}$ for all $j\in[k]$. This means there exists a partition of $M$ into $n$ bundles where $v^e_i(P^i_j) \geq \frac{\mu_i}{\beta}$. By the definition of MMS, this implies $\mu^e_i \geq \frac{\mu_i}{\beta}$. Doing a similar analysis starting from the MMS partition of type $i$ with respect to $\I^{e}$, we obtain $\beta \mu_i \geq \mu^e_i$. Hence, 
    \begin{equation}\label{eq:mu_is}
        \beta \mu_i \geq \mu^e_i \geq \frac{\mu_i}{\beta}.
    \end{equation} 
    
    Since the algorithm is $\alpha-$MMS competitive on $\I^{e}$, $v^e_i(b) \geq \alpha \cdot \mu^{e}_i$. Using \cref{eq:mu_is}, and considering linearity of valuations, $\beta v_i(b) \geq v^e_i(b)$, we obtain $\beta v_i(b) \geq \alpha \frac{\mu_i}{\beta}$. Hence, $v_i(b) \geq \frac{\alpha}{\beta^2} \mu_i$. Since this analysis holds for all agents, every agent receives a bundle she values at least $\frac{\alpha}{\beta^2}$ times her MMS value with respect to $\I$. 
\end{proof}

The lemma above implies that in a learning-augmented framework—where the \(k\) valuation functions are treated as predictions subject to multiplicative errors (with the erroneous input \(\I^{e}\) used by our algorithms in place of the ground truth instance \(\I\))—the performance guarantees of our algorithms degrade gracefully in proportion to the multiplicative error factor.
\section{Adversarial Arrival of Agents}\label{sec:adversary}
In this section, we study the $\problem$ problem in the adversarial arrival model. We design an algorithm that is $\frac{1}{k}$-MMS competitive, and prove a lower bound to show that there is no online algorithm that can give an approximation better than $\frac{{2}}{\sqrt{k}-2}$, where $k$ is the number of types. Formally, we prove the following two theorems in this section.
\thmadvpos*

\thmadvneg*
\subsection{A $\frac{1}{k}$-MMS-competitive algorithm}
In this section, we will prove \Cref{thm:adv}. Before proceeding, we introduce further notation relevant to this section. 

\noindent \textit{Tentative Allocations.} For each type \(i \in [k]\), let \(G_i\) denote the collection of bundles tentatively allocated to agents of type \(i\). A bundle is tentatively allocated for type \(i\) only if it is valued at least \(\frac{1}{k}\) by that type. Note that the bundles within a given \(G_i\) are disjoint, but a single bundle may be tentatively allocated for multiple types.

\noindent \textit{Saturated types.} We call a type \emph{saturated} if the set $G_i$ for this type has $n-t$ bundles in it, right before the arrival of the $(t+1)$th agent. In the algorithm we use the set of unsaturated types which is the set of all types minus the set of saturated types and denote it using $\mathtt{unsaturated}$. 

Finally, we maintain $R$ as the set of items that are not allocated by the algorithm in any previous round and $P$ is the pool of available items which excludes both tentatively and previously allocated items. As always, we assume the input instance is normalized, so the MMS value of each type is $1$.

\subsubsection{Overview of Algorithm} 

The algorithm operates in two phases. In the pre-processing phase, for each agent type \(i \in [k]\), up to \(n\) items that are valued above \(\frac{1}{k}\) by type \(i\) are identified. Each such item is put into a singleton bundle and tentatively added to the collection \(G_i\). When an agent \(t\) arrives and reveals her type \(i\), the algorithm first checks whether \(G_i\) is non-empty. If it is, an arbitrary bundle from \(G_i\) is allocated to agent \(t\) and removed from every \(G_j\) that contains it. 
If \(G_i\) is empty, a bag-filling procedure is initiated using the pool of available items. Items are sequentially added to a bag until at least one unsaturated type values the bag at least \(\frac{1}{k}\). If the arriving agent’s type \(i\) values the bag at least \(\frac{1}{k}\), the bag is allocated to her immediately. Otherwise, the bag is tentatively allocated to all unsaturated types that value it at least \(\frac{1}{k}\); any type that becomes saturated as a result triggers an update of the unsaturated set, and the items in the bag are removed from the available pool. This bag-filling process repeats until a bundle is allocated to agent \(t\). Finally, if any type possesses more tentative bundles than the number of remaining agents, extra bundles are arbitrarily removed.

\begin{algorithm}[h!]
    \caption{$\problem$ problem in adversarial model}
    \label{alg:adversarial-alg}
    \SetKwInOut{Input}{Input}\SetKwInOut{Output}{Output}
  \Input{An instance $\I= (N,M,\{v_i\}^k_{i=1})$}
  \Output{An allocation of items among agents such that each agent receives a value of at least $\frac{1}{k}$.}
  \BlankLine
  Initialize $R \gets [m]$ and let $n = |N|$ \tcp*{The set of unallocated items}
  $(G_i)_{i \in [k]} \gets$ Preprocess($\I$) \label{adv-step:preprocess}\tcp*{Preprocess to handle high-valued items} 
  For all types that have $|G_i| < n$ add them to saturated types i.e., $\unsat \gets \{i \mid |G_i| < n\}$\;
  \For{$t = 1$ \KwTo $n$}{
  Agent $t$ arrives and reveals type $i$\;
  \eIf{$G_{i} \neq \emptyset$}{
    Allocate to this agent any bundle from her $G_{i}$ and remove this bundle from $G_j$ for all $j \in [k]$. \label{adv-step:allocate-reserved}
  }{
    $P \gets R \setminus \{\cup_{j \in [k]} \cup_{g \in G_j}g\}$ \tcp*{Pool of available items}
    \While{$\sf true$}{
            Initialize Bag $B \gets \emptyset$ \; 
            Fill $B$ arbitrarily with items from $P$ till at least for one type $j \in \unsat$, $v_j(B)\geq \frac{1}{k}$\label{adv-step:bag-fill}\;
            \eIf{$v_{i}(B) \geq \frac{1}{k}$} {\label{adv-step:allocate-start}
                Allocate B to agent $t$\;
                Update $R \gets R \setminus B$ \;
                break out of loop \label{adv-step:allocate-end}
            }{
                \For{all unsaturated types $j \in \unsat$ such that $v_j(B) \geq \frac{1}{k}$}{ \label{adv-step:other-reserve-start}
                    Update $G_j \gets G_j \cup B$ \tcp*{Tentatively allocate $B$}
                    \If{$|G_j| \geq n-t$}{update $\unsat \gets \unsat \setminus \{j\}$}
                    Update $P \gets P \setminus B$ \tcp*{Update the pool of available items} \label{adv-step:other-reserve-end}
                }
            }
            }
        }
    For any type such that $|G_j| > n-t$, release arbitrary $|G_j| - (n-t)$ bundles from her $G_j$. \label{adv-step:maintain-size}
  }
  
\end{algorithm}
\begin{algorithm}[h!]
\caption{\textbf{Preprocessing}: Handling large valued items}
    \label{alg:preprocessing}
    \SetKwInOut{Input}{Input}\SetKwInOut{Output}{Output}
  \Input{An instance $\I= (N,M,\{v_i\}^k_{i=1})$}
  \Output{$\{(G_i)_{i \in [k]}\}$ where $G_i$ is the tentative set of bundles allocated to type $i$}
  \BlankLine
Initialize $G_i \gets \emptyset$ $\forall i\in [k]$\ \\
\For{$i = 1$ \KwTo $k$}{
    \For{each $g \in M$}{
        \If{$v_i(g) \geq \frac{1}{k}$ and $|G_i| < |N|$}{
            $G_i \gets G_i \cup \{\{g\}\}$ \
        }
    }
}
\end{algorithm}
\subsubsection{Analysis}
The central idea of the proof is to show that the Algorithm maintains the following invariant. 
\begin{invariant}[label=inv:adv]
\textbf{Invariant:} Right before the arrival of agent $t+1$, for each type $i \in [k]$, either $|G_i| = n-t$ or $v_i(R) \geq n-t$ holds.
\end{invariant}
\noindent This invariant states that for each type, we either maintain a tentative allocation ensuring that, even if all remaining agents belong to that type, there are sufficient bundles to allocate one to each agent, or we retain a sufficiently large remaining value in \( R \) to fill a bag for an arriving agent of that type. This guarantees that every agent receives a value of at least \( \frac{1}{k} \) upon arrival. Towards this, we first establish some helpful claims.
\begin{claim}\label{cl:atmost1/k}
    After the preprocessing step, for every unsaturated type $i\in[k]$, and item $g\in P$, $v_i(g)< \frac{1}{k}$.
\end{claim}
\begin{proof}
    Suppose for contradiction that there is an item $g\in P$, where $v_i(g)\geq \frac{1}{k}$. As type $i$ is not saturated, $|G_i| < n$, hence this item should have been added as a singleton bundle to $G_i$ in the preprocessing step. As $P$ excludes all items tentatively allocated to all types, $g\notin P$, and this contradicts with our assumption.
\end{proof}

\begin{claim}\label{clm:adv-saturated}
    At any point if a type gets saturated in the algorithm then she remains saturated till the end.
\end{claim}
\begin{proof}
    Recall that we call a type $i\in[k]$ saturated when we have $r$ agents yet to come and the type has $r$ tentatively allocated bundles. Right before agent $t$ arrives, if type $i$ is saturated that means $G_i$ has $n - (t-1)$ tentatively allocated bundles. When agent $t$ comes in and takes a bundle, she will either take a bundle that also belongs to $G_i$ or she will take some other bundle. In the first case, we still have $n-t$ tentative bundles reserved for type $i$, and in the second case we will drop a bundle to avoid saving extra bundles. In both cases, before the arrival of agent $t+1$, $|G_i| = n-t$, hence the type remains saturated.
\end{proof}

\begin{claim}\label{clm:adv-bundle-value}
    At any point in the algorithm, for any unsaturated type $i \in [k]$, and a tentatively allocated bundle $B \in \bigcup_{j\in[k]} G_j$, we have $v_i(B)\leq 1$. 
\end{claim}
\begin{proof}
If $B$ contains a single item, as the input is normalized, the value of every item is bounded by $1$ for all types. Hence, $v_i(B)\leq 1$.

If $B$ contains multiple items, it was created during the bag-filling procedure from the pool of remained items. We will prove that if $B \in G_i$, then $v_i(B) \leq \frac{2}{k}$. If $B \notin G_i$, then $v_i(B) < \frac{1}{k}$. In both scenarios as $k\geq 2$, we have that $v_i(B) \leq 1$.

Consider a bundle $B$ that is tentatively allocated in the bag-filling procedure. This bundle is assigned to any unsaturated type \( i \in [k] \) that values it at least \( \frac{1}{k} \).

Thus, if \( B \) is not allocated to an unsaturated type, its value for that type must be at most \( \frac{1}{k} \). Conversely, whenever an unsaturated type values \( B \) at least \( \frac{1}{k} \), it is tentatively allocated to her. Right before adding the last item to \( B \), no unsaturated type valued it at \( \frac{1}{k} \). Since all items in the pool have a value of at most \( \frac{1}{k} \) for all unsaturated types by \cref{cl:atmost1/k}, the last item added to \( B \) also has a value of at most \( \frac{1}{k} \) for them. Together, these observations imply that the total value of \( B \) can be at most \( \frac{2}{k} \) for the types that received $B$ as a tentative allocation.
\end{proof}

We can now prove the following main lemma.
\begin{lemma}\label{lem:invariant-adv}
    Algorithm \ref{alg:adversarial-alg} maintains Invariant \ref{inv:adv} throughout the algorithm.
\end{lemma}
\begin{proof}[Proof of Lemma \ref{lem:invariant-adv}]
    We prove this using induction on the number of agents that have already arrived. Before any agent arrives, all items are still unallocated and therefore for all types $i \in [k]$, we have $v_i(R) = v_i([m]) = n$ and the invariant holds. 

    Now consider any $t \in [n]$. Assuming that the invariant holds before the arrival of agent $t$, we prove that it holds before the arrival of agent $t+1$ as well. Fix any type $i \in [k]$. If $i$ is saturated before the arrival of agent $t$, then by Claim \ref{clm:adv-saturated}, $i$ remains saturated and the invariant holds. Therefore, from now on suppose $i$ is unsaturated before the arrival of agent $t$. Since the invariant holds, we have $v_i(R) \geq n-(t-1)$ before allocating any items to agent $t$. Let the type of agent $t$ be $j$. Consider the following cases.
    \paragraph{Case 1: $i = j$.} If the current agent is of type $i$, and if $G_i \neq \emptyset$, then we allocate a bundle from $G_i$ to agent $t$. Using claim \ref{clm:adv-bundle-value}, the value of this bundle is at most $1$ for type $i$. On the other hand if $G_i = \emptyset$ then she takes a bundle $B$ via bag-filling from the pool of remained items. Right before adding the last item to the bundle, the bundle was unclaimed (valued less than $\frac{1}{k}$) by all unsaturated types, including $i$. By \cref{cl:atmost1/k}, the last added item can have a value of at most $\frac{1}{k}$ for all unsaturated types; hence the value of $B$ can be at most $\frac{2}{k}$, and since $k\geq 2$, $v_i(B)\leq 1.$

    \paragraph{Case 2: $i \neq j$.} Now suppose agent $t$ is of type $j\neq i$. 
    \begin{enumerate}
        \item  If agent $t$ takes a bundle $B\in G_j$, by \cref{clm:adv-bundle-value}, $v_i(B)\leq 1$.
        \item If $G_j = \emptyset$ when agent $t$ arrives, she takes a bundle $B$ via bag-filling from the remained pool of items. If this bundle is never claimed by type $i$, then $v_i(B) < \frac{1}{k}$. If the bundle is claimed by type $i$, it is only allocated to agent $t$ if both types $i$ and $j$ claim it at the same time. In this case, since no unsaturated type claimed the bundle before addition of the last item, and the last added item has a value of at most $\frac{1}{k}$ for all unsaturated types by \cref{cl:atmost1/k}, the bundle's value can be at most $\frac{2}{k} \leq 1$ since we assume $k \geq 2$. Hence, even in this case, $v_i(B)\leq1$.
    \end{enumerate}
    We proved that in all cases, the value type $i$ loses to agent $t$ is at most $1$. Therefore, the remained value in $R$ after allocating this bundle is at least $n - (t - 1) - 1 \geq n-t$. Hence the invariant holds.
    \end{proof}

    Now we prove our next main lemma which says that as long as the invariant is true, each incoming agent can be allocated a bundle of value $\frac{1}{k}$.

\begin{lemma}\label{lem:adv-1/k}
    If the invariant is maintained, each agent $t \in N$ can be allocated a value of $\frac{1}{k}$.
\end{lemma}
\begin{proof}
    When agent $t$ comes in, suppose her type is $i$. If $G_{i} \neq \emptyset$ then we can assign her a bundle from the tentative allocation $G_i$. Note that a bundle $B$ is only added to $G_i$ if $v_i(B)\geq \frac{1}{k}$. Hence, the agent receives a bundle with value of at least $\frac{1}{k}$. 
    
     Otherwise, if \(G_i = \emptyset\), the algorithm attempts to fill a new bag until its value for type \(i\) reaches at least \(\frac{1}{k}\). During this bag-filling process, some bundles may be tentatively allocated to other types. We want to show that despite this, agent $t$ eventually receives a bag valued at least \(\frac{1}{k}\). To see this, note that if agent $t$ (or equivalently type $i$) claims a bag with value at least $\frac{1}{k}$, the bag is immediately allocated to the agent. The only way she loses value to other types during this process is if the other types claim the bag while she does not. However, in such cases, each bag tentatively allocated to other types has a value of less than \(\frac{1}{k}\) for type \(i\). The maximum number of such bags is \((n - (t-1))\cdot(k-1)\). This follows because, for each of the other \(k-1\) types, their tentative allocations can include at most \(n - (t-1)\) bags—beyond this point, the agent is saturated, and no further bundles are allocated to it\footnote{A careful reader might note that we need to only keep $(n-t)$ bags to saturate any agent, however this does not give us any gain in the approximation factor.}. Therefore, 
    \begin{align*}
        v_i(\cup_{j \in [k], j \neq i} G_j) \leq (n-(t-1))\cdot (k-1) \frac{1}{k}
    \end{align*}
    Moreover, since $G_i = \emptyset$, $v_i(G_i) = 0$. Hence, $$v_i(\cup_{j \in [k]} G_j) \leq (n-(t-1))\cdot (k-1) \frac{1}{k}.$$
    
    Given the pool of items $P = R \setminus \{\cup_{j \in [k]} \cup_{g \in G_j}g\}$, by additivity of valuations, we have $v_i(P) = v_i(R) - v_i(\bigcup_{j\in[k]} G_j)$. By our invariant, $v_i(R) \geq n-t + 1$. Putting these together, the remained value in the pool of items $P$ satisfies
    \begin{align*}
        v_i(P) = v_i(R) - v_i(\cup_{j \in [k]} G_j)  \geq (n-t+1)\cdot(1 - (k-1)\frac{1}{k}) \geq (n-t+1)\frac{1}{k}.
    \end{align*}
    Therefore, as long as \( t \leq n \), a sufficient amount of value remains in the pool of remaining items. Consequently, we can still construct a bag with a value of at least \( \frac{1}{k} \) from \( P \), even after multiple bundles have been tentatively allocated to other types.
\end{proof}

Lemmas \ref{lem:invariant-adv} and \ref{lem:adv-1/k} together establish \Cref{thm:adv}.

\subsection{Non-existence of $\Omega(\frac{1}{\sqrt{k}})$-MMS-competitive algorithm}\label{subsec:adv-lower-bound}

We now complement our algorithmic result with a proof for non-existence of a ${\frac{2}{{\sqrt{k} - 2}}}$-MMS competitive algorithm under adversarial arrivals. We note that this holds even when $k$ is a constant. 

For this result, we will assume that the valuations are not normalized. This lets us keep our valuations binary i.e., all types value all items either $0$ or $1$. The same example works with normalized valuations. However, it is beneficial to work with binary valuations to notice the stark contrast with the offline setting where a polynomial-time computable exact MMS allocation is known to exist.
\thmadvneg*

\begin{proof}
    For any $k$, and even number $n$, we construct an instance with binary valuations as follows. Let $\mu_1 = \Big\lfloor {\frac{\sqrt{k}}{2}} \Big\rfloor$. Note that $k \geq 4\binom{\mu_1}{2} + \mu_1 + 1$. We describe the valuation function of $4\binom{\mu_1}{2} + \mu_1 + 1$ types, and the remained $k-(4\binom{\mu_1}{2} + \mu_1 + 1)$ types can have any arbitrary valuations.  

    Our instance has $n \mu_1$ items. Let us denote these items by $\{1, 2, \ldots, n\mu_1\}$. Let us partition these items into $\mu_1$ intervals, each containing $n$ items. Interval $r\in[\mu_1]$ is $G_r = \{(r-1)n +1, rn\}$.
    
    The first type, $i$, values all of the items at $1$. Therefore, this type's MMS value is $\mu_1$. For all the remaining types, the MMS value is $1$ and therefore they value exactly $n$ items at $1$, and the rest of the items at $0$. Call $\mu_1$ of these types $i_1, \ldots i_{\mu_1}$. For any $r\in[\mu_1]$, type $i_r$ values items in $G_r$ at $1$ and all others at $0$. 
    
    Finally, we define the remaining $4 \binom{\mu_1}{2}$ types. These types are created so that we have four types associated with every pair of distinct types from $\{i_1, \ldots, i_{\mu_1}\}$. For any pair of distinct $l, r \in [\mu_1]$, we define four new types, denoted as $\{i^{1}_{l,r}, i^{2}_{l,r}, i^{3}_{l,r}, i^{4}_{l,r}\}$. Now to see how the valuations of these types are defined, fix a pair $\{l,r\}$. Given this pair, to define the valuation for $i^p_{l,r}$ for every $p\in[4]$, arbitrarily divide $G_l$ and $G_r$ into two parts of equal size. Call these $G_l^1, G_l^2$ and $G_r^1, G_r^2$. Now we define the valuations for the four types corresponding to pair $\{l,r\}$ as follows:
    \begin{itemize}[leftmargin=*]
        \item $i^1_{l,r}$ likes items in $G_l^1 \cup G_r^1$ at $1$ and all others at $0$.
        \item $i^2_{l,r}$ likes items in $G_l^1 \cup G_r^2$ at $1$ and all others at $0$.
        \item $i^3_{l,r}$ likes items in $G_l^2 \cup G_r^1$ at $1$ and all others at $0$.
        \item $i^4_{l,r}$ likes items in $G_l^2 \cup G_r^2$ at $1$ and all others at $0$.
    \end{itemize}
    Consider an adversarial order where the first agent is of type $i$. To give this agent a better than $ {\frac{2}{\sqrt{k} - 2}} = \frac{1}{\frac{\sqrt{k}}{2} - 1} \geq \frac{1}{\mu_1}$ MMS approximation, we must give this agent at least two items. Now if both these items belong to same interval $G_r$ for some $r \in [\mu_1]$, then type $i_r$ looses two items to the first agent, and is remained with $n-2$ valuable items. Then, the adversary can send all the $n-1$ remaining agents of type $i_r$ and therefore the last agent will get a value of $0$. To avoid this, the algorithm must choose two different intervals $G_r$, and $G_l$ and pick one item from each, for assigning them to the first agent. Note that for any two items chosen from $G_l$ and $G_r$, there exists one type $i^p_{l,r}$ where $p \in [4]$ that values both items at $1$. That type loses two valuable items to the first agent, hence is remained with only $n-2$ valuable items. Hence, the adversary can send all the $n-1$ remaining agents of this type, so the last agent has no valuable item left and gets a value of $0$. Consequently, the algorithm cannot select any pair of items for the first agent that remains robust against the adversary's decisions. Thus, no online algorithm can guarantee better than a $\frac{2}{\sqrt{k}-2}$-MMS. 
\end{proof}
\section{Stochastic Arrival of Agents with Known Distribution}\label{sec:knownD}
In this section, we study the \textsc{OnlineKTypeFD} problem under the stochastic arrival model, where agents arrive according to a known probability distribution. 
Under mild assumptions, we design an algorithm that is \(\alpha\)-MMS competitive, with \(\alpha\) approaching \(\frac{1}{2}\) as \(n\) grows large.

For each type \(i \in [k]\), let \(X^n_i\) denote the number of type \(i\) agents observed ex-post after all \(n\) agents have arrived. Although \(X^n_i\) is unknown until the final agent arrives, the expected value \(\mathbb{E}(X^n_i) = n p_i\) is known in advance, where \(p_i\) is the probability that an arriving agent is of type \(i\). Since each agent belongs to exactly one type, it follows that \(\sum_{i=1}^k X^n_i = n\). The following proposition follows directly from Chernoff bound.

\begin{proposition}\label{prop:boundx_i} Consider an \textsc{OnlineKTypeFD} problem on input $(\I= (N,M,V=\{v_i\}^k_{i=1}), D)$, where $n=|N|$. For any $\epsilon>0$, and any type $i\in[k]$, the deviation of $X^n_i$ from its expectation satisfies the following $P\left[|X^n_i - np_i| \leq \Theta(n^{\epsilon} \sqrt{np_i})\right] \geq 1 - O(e^{-n^{2\epsilon}})$.
\end{proposition}

In this section, whenever we use $\{M',C\}$, we refer to the set of items $M = M' \cup C$, with the distinction that elements of $M'$ and $C$ are considered as two separate lists. For every lemma from this section that is referenced in \cref{sec:unknowndist}, we denote the number of agents as \( |N| \) rather than \( n \). Although we assume that \( n = |N| \) in the input, this notation is adopted deliberately to maintain generality for next section. Recall that given the desired MMS approximation factor $\alpha$, all items that are valued at more than $\alpha$ by all types are called \emph{universally high-valued items}. The rest of this section is organized as follows. In \Cref{subsec:correctpredicchernoff} we present the main algorithm (\Cref{alg:knownD}) that handles all instances. This algorithm builds upon two algorithms: \Cref{alg:highC} that handles instances with large number of universally high-valued items and \Cref{alg:lowC} that is invoked when there are no universally high-valued items. These algorithms are described in \Cref{sec:highC} and \Cref{sec:lowC} respectively. Finally, in \Cref{app:tightness} we show that the bounds obtained in this theorem are almost-tight for our algorithm.

\subsection{An (Almost) $\frac{1}{2}$-MMS-competitive Algorithm}\label{subsec:correctpredicchernoff}
\Cref{alg:knownD} is our main algorithm for this model. Given the input parameter \(\alpha\), \cref{alg:knownD} first identifies the set \(C\) of universally high-valued items. Since any singleton bundle containing an item from \(C\) is valued by more than $\alpha$ for all types, such bundles can be allocated to any arriving agent. If the number of items in \(C\) exceeds the threshold specified in line \ref{preproc:Cgeq}, the algorithm calls \cref{alg:highC}. Otherwise, each item in \(C\) is allocated to the first \(|C|\) arriving agents, reducing the problem to a smaller instance \(\mathcal{I}' = ([|C|+1, n], M \setminus C, \{v_i\}_{i=1}^k)\) where no item is valued above \(\alpha\) by all \(k\) types. The algorithm then proceeds with \cref{alg:lowC} on this reduced instance.

We will describe each of these subroutines in detail and demonstrate that, in both cases, an $\alpha-$MMS guarantee is ensured when $n$ is asymptotically large, where $\alpha$ can be made arbitrarily close to $\frac{1}{2}$. The main result of this section is the following theorem.
\thmknownD*
\noindent We establish \Cref{thm:knownD} by proving following two lemmas. The first lemma, proved in \Cref{sec:highC} says that our algorithm can handle cases when we have large number of universally high-valued items. The second lemma, proved in \Cref{sec:lowC} says that our algorithm can handle the case when we have a small number of universally high-valued items.

\begin{restatable}{lemma}{lemfirstimp}\label{lemma:firstimp}
    When $|C| \geq n(1-\frac{1}{k}) + n^{\epsilon}\sqrt{np_1}$, for any $\alpha\leq 1$, \cref{alg:knownD} is $\alpha-$MMS competitive with probability at least $1-O(e^{-n^{2\epsilon}})$. 
\end{restatable}

\begin{restatable}{lemma}{lemsecondimp}\label{lem:secondimp}
    If \( |C| < n\left(1 - \frac{1}{k}\right) + n^{\epsilon}\sqrt{n p_1} \), $p_{k} = \omega(\frac{n^\epsilon k}{\sqrt{n}})$, for any small $\eta>0$, there exists an $n(\eta)$, where for all $n>n(\eta)$, \cref{alg:knownD} is \( \frac{1}{2(1+\eta)}-\)MMS competitive with probability at least \( 1 - o(e^{-n^{1.5\epsilon}}) \).
\end{restatable}

\begin{algorithm}[ht!]
\caption{Stochastic Arrival of Agents with Known Distribution}\label{alg:knownD}
    \SetKwInOut{Input}{Input}\SetKwInOut{Output}{Output}
  \Input{Stochastic Instance $\I = (N,M, \{v_i\}_{i\in [k]})$, distribution $D$ with $p_1\geq p_2\geq\cdots\geq p_k$, $\alpha = \frac{1}{2(1+\eta)}$, $\epsilon$.}
  \Output{ $\alpha$-MMS.}
  \BlankLine
Let $n = |N|$, $C = \{g\in M \mid \forall i\in [k] :v_i(g)\geq \alpha \}$, $M' = M\setminus C$.\label{preproc:C}\\
\eIf{$|C|\geq n(1-\frac{1}{k}) + n^{\epsilon}\sqrt{np_1}$\label{preproc:Cgeq}}{
Let $\Icalhat = (N, \{M',C\}, \{v_i\}^k_{i=1}).$\\
Run \cref{alg:highC} on $(\Icalhat, D,\alpha)$.
}
{
\For{$t \in [|C|]$\label{alg5:onlinestart}}{
    Agent $t$ arrives and reveals type $j$.\\
    Let $g$ be an arbitrary item in $C$.
    Assign $\{g\}$ to agent $t$, and remove it from $C$. 
}\label{alg5:onlineend}
Let $\Ical' = ([|C|+1, n],M', \{v_i\}^k_{i=1})$ be the reduced instance. \\
Run \cref{alg:lowC} on $(\Ical', D, \epsilon, \alpha)$.
}
\end{algorithm}

\subsection{Large Number of Universally High-Valued Items.}\label{sec:highC}
\textbf{Description of \cref{alg:highC}}  
The algorithm operates in two phases. In the preprocessing phase, the most frequent type (type 1) is asked to compute an \(\alpha\)-MMS partition \(A\). The algorithm then adjusts \(A\) so that each bundle contains exactly one item from the set \(C\) of universally high-valued items: any bundle with more than one such item has its extra items removed and each extra high-valued item is placed in its own singleton bundle, which is then added to \(A\). Additionally, bundles in \(A\) that contain no item from \(C\) are removed and stored in a reserved set \(G_1\) (subject to \(|G_1| \le n - |C|\)), thereby ensuring that \(A\) reserves bundles for all types while maintaining a separate set of reserved bundles for type 1 in $G_1$. In the online phase, when agents arrive sequentially, a type 1 agent is allocated a bundle from \(G_1\) if available; otherwise, they receive a bundle from \(A\). Agents of other types are allocated bundles from \(A\). 
\begin{algorithm}[ht!]
\caption{Large number of universally high-valued items}\label{alg:highC}
\SetKwInOut{Input}{Input}\SetKwInOut{Output}{Output}
\Input{A Stochastic Instance $\I = (N, \{M',C\},\{v_i\}^k_{i=1})$, $D$ with $p_1\geq p_2\geq\cdots\geq p_k$, $\alpha$.}
\Output{$\alpha$-MMS.}
\BlankLine
$G_1 = \emptyset$, $A\gets \alpha$-MMS partition of type with the highest probability, in $\I$. \label{highC:prepstart}\\
\For{$r\in N$}{
  $x_r = |A_r\cap C|$.\\
    \If{$x_r \geq 2$}{Remove $x_r-1$ items of $A_r\cap C$ from $A_r$ and put each item in a separate new bag, then add the bag to $A$.}
    \If{$x_r = 0$}{$A \gets A\setminus A_r$.\\
        \If{$|G_1|\leq |N|-|C|$}{
        $G_1 \gets G_1 \cup A_r$.
        }
    }
}\label{highC:prepend}
\For{$t \in N$\label{highC:onlinestart}}{
        Agent $t$ arrives and reveals type $j$.\\
        \eIf{$j==1$ and $G_1\neq \emptyset$} {
                Pick a bag $b$ from $G_1$ and assign it to agent $t$.\\
                $G_1 \gets G_1\setminus b$.
            }
         {
            Pick a bag $b$ from $A$ and assign it to agent $t$.\\
                $A \gets A\setminus b$.\label{alg-stoch:endfirst}
         }   
    }\label{highC:onlineend}
\end{algorithm}

\subsubsection{Analysis of \cref{alg:highC}}\label{app:highC}

First, let us prove the following useful lemma. 
\begin{lemma}\label{lem:highC}
    Given a stochastic instance $\I = (N, \{M',C\}, \{v_i\}^k_{i=1})$ and distribution $D$, if the number of type $1$ agents arrived ex-post $X^{|N|}_1$ satisfies $X^{|N|}_1 \geq |N| - |C|$ for any $\alpha \leq 1$, \cref{alg:highC} is $\alpha-$MMS competitive. 
\end{lemma}
\begin{proof}
During preprocessing, since any bundles removed from \( A \) do not contain items from \( C \), all items in \( C \) remain in \( A \). Given that each bundle in \( A \) contains exactly one item from \( C \), it follows that \( |A| = |C| \). Moreover, since every bundle in \( A \) has a value greater than \( \alpha \) for all agent types, any bundle of it can be allocated to any arriving agent. 
We analyze two cases separately: when $|C| \geq |N|$, and when $|N|-1 \geq |C|$.

\begin{enumerate}
    \item If $|C| \geq |N|$, then $|N| - |C|\leq 0$, implying that $G_1$ remains empty during the preprocessing step. Consequently, in the online phase, each arriving agent is allocated a bundle from $A$. Since $|A| = |C| \geq |N|$, the number of available bundles in $A$ is sufficient to provide each agent with a bundle valued at more than $\alpha$.

    \item If $|N|-1 \geq |C|$, the bundles in $G_1$ belong to the $\alpha$-$\MMS$ partition of type $1$ and are therefore claimed by type 1 agents. Hence, all bundles in $G_1$ are exclusively reserved for type $1$ agents. Since there are $|C|$ items in $C$, at least $|N| - |C|$ bundles in the $\alpha$-$\MMS$ partition of type $1$ contain no items from $C$. This implies that at the end of preprocessing, $|G_1| = |N| - |C|$. Hence, before any agent arrives, the total number of reserved bundles satisfies $|G_1| + |A| = |N|$. During the online phase, each arriving agent receives one bundle from either $A$ or $G_1$, ensuring that after the $t$-th arrival, $|G_1| + |A| = |N| - t$.

     By assumption, $X^{|N|}_1 \geq |N| - |C|$. Therefore, the number of arriving type $1$ agents exceeds the number of bundles reserved for them in $G_1$. As a result, at some point after the arrival of a type $1$ agent, $G_1$ becomes empty, and from that moment onward, all type $1$ agents will receive bundles from $A$. Let $t^*$ denote the index of the last agent who received a bundle from $G_1$.

    Prior to the depletion of $G_1$, only agents of types in $[2,k]$ received bundles from $A$. Given that $\sum_{i=1}^{k} X^{|N|}_i = |N|$ and the bound $X^{|N|}_1 \geq |N| - |C|$, it follows that $\sum_{i=2}^{k} X^{|N|}_i \leq |C|$. Therefore, the number of initial bundles in $A$ is sufficient for all agents of types $[2,k]$ who arrived before the $t^*$-th agent. Additionally, before the arrival of the $t^*$-th agent, $G_1 \neq \emptyset$, ensuring that any type $1$ agent can still receive a bundle from $G_1$. 

    Once agent $t^*+1$ arrives, as $G_1$ is depleted and $|G_1| + |A| = |N| - t^*$, it implies that $|A| = |N| - t^*$. Consequently, the remaining number of bundles in $A$ matches the number of remaining agents. Since each bundle in $A$ contains one item from $C$ and is claimed by all agent types, each bundle can be assigned to a unique remaining agent.
\end{enumerate}
\end{proof}

We are now ready to prove the main lemma of this Section. 

\lemfirstimp*
\begin{proof}
According to \cref{prop:boundx_i}, with probability at least \( 1 - O(e^{-n^{2\epsilon}}) \), we have 
\[
X^n_1 \geq np_1 - n^{\epsilon} \sqrt{np_1}.
\]
Since type \( 1 \) is the most frequent, it follows that \( p_1 \geq \frac{1}{k} \), which implies that, with high probability,
\[
X^n_1 \geq \frac{n}{k} - n^{\epsilon} \sqrt{np_1}.
\]
Given the condition \( |C| \geq n(1 - \frac{1}{k}) + n^{\epsilon} \sqrt{np_1} \), we conclude that, with high probability, \( X^n_1 \geq n - |C| \). Applying \cref{lem:highC}, the desired result follows.
\end{proof}

\subsection{No Universally High-Valued Items}\label{sec:lowC}
\textbf{Description of \cref{alg:lowC}} 
Our algorithm begins with a preprocessing phase for a reduced instance \(\Ical'=(N',M',\{v_i\}_{i\in [k]})\) with \(n' = |N'|\). For each type \(i\), we first define \(T_i\) as the set of high-valued items, noting that \(\bigcap_{i=1}^k T_i = \emptyset\) since there is no universally high-valued item in a reduced instance. We then compute the expected number of agents \(\mu_i = n' p_i\) for each type and prepare a reserved collection \(G_i\) of bundles. We also compute $M_i = \lfloor\mu_i + n'^{\epsilon}\sqrt{\mu_i}\rfloor$ as the upper limit on the size of $G_i$. A type $i$ is saturated if $|G_i| = M_i$. An ordering \(L\) of types is determined based on whether a pair \(i,j\) exists with \(|T_i \setminus T_j| \ge \frac{n' p_i}{2}\); if so, type \(i\) is ordered first and \(j\) last, otherwise the least frequent type is placed last. Adhering to the ordering \(L\), high-valued items are allocated sequentially. When it is the turn of type \( i \), as long as \( |G_i| < M_i \), an available high-valued item from \( T_i \setminus \hat{T_i} \) (where $\hat{T}_i$ records items that have been allocated) is selected, prioritizing items that are not highly valued by the last type in the order. The chosen item forms a new singleton bundle, which is then added to \( G_i \). Then, we iteratively construct additional bundles from the remaining items \(R\) until every unsaturated type is saturated, or no item in $R$ remains. When an agent arrives and reveals her type \(j\), she is immediately assigned a bundle from the pre-prepared collection \(G_j\). 

\begin{algorithm}[h!]
\caption{No universally high-valued items}\label{alg:lowC}
\SetKwInOut{Input}{Input}\SetKwInOut{Output}{Output}
\Input{A Stochastic Instance $\I' = (N', M', \{v_i\}_{i\in [k]})$, $D$ with $p_1\geq p_2\geq\cdots\geq p_k$, $\epsilon, \alpha$.}
\Output{$\alpha\text{-MMS}$}
\BlankLine
$n' \gets |N'|$. \label{lowC:prepstart}\\
For all $i\in[k]$, let $T_i = \{g\in M' \mid v_i(g)\geq \alpha\}$, $\mu_i = n'p_i$,  $G_i = \emptyset$, $M_i = \lfloor \mu_i +n'^{\epsilon}\sqrt{\mu_i} \rfloor$, $\hat{T_i} = \emptyset$, .\label{preproc:T_i}\\
$\mathtt{unsaturated} = [k]$.\label{defineend}\\
    \eIf{$\exists i,j\in[k]$ s.t. $|T_i\setminus T_j|\geq \frac{n'p_i}{2}$\label{sortstart}}{
    Let $L$ be a list of an arbitrary ordering of $[k]$ where $i$ is first and $j$ is last.
    }{
    Let $L$ be a list of an arbitrary ordering of $[k]$ where the type with the lowest probability is last.\label{lastshouldbek}
    }\label{sortend}
Let $k' = L[k]$. \\
\For{$r\in [k]$\label{turnstart}}{
    Let $i = L[r]$.\\
    \While{$T_i \setminus \hat{T_i}\neq \emptyset$ and $|G_i| < M_i$\label{whilehighval}}{
        If $T_i \setminus (\hat{T_i}\cup T_{k'})\neq \emptyset$, pick an item $g$ from $T_i \setminus (\hat{T_i}\cup T_{k'})$, otherwise pick any item from $T_i \setminus \hat{T_i}$. \label{prioritizing}\\
        $G_i \gets G_i \cup \{g\}$ and $M' \gets M'\setminus g$.\\
        \For{$j \in [r,k]$}{
            $\hat{T}_{L[j]} \gets \hat{T}_{L[j]} \cup g$.
        }
    }
    If $|G_i| = M_i$, Remove $i$ from $\mathtt{unsaturated}$.
}\label{turnend}
$R \gets M'$.\label{BFstart}\\
\While{$\mathtt{unsaturated \neq \emptyset}$ and $R\neq \emptyset$\label{while:BF}}{
Initialize Bag $B\gets \emptyset$.
Fill $B$ arbitrarily with items from $R$ till at least one $\mathtt{unsaturated}$ type claims $B$. \\
Pick an arbitrary $\mathtt{unsaturated}$ type $i$ that claimed the bag and add $B$ to $G_i$. \\
Update $R\gets R \setminus B$. \\
If $|G_i| = M_i$, remove $i$ from $\mathtt{unsaturated}$. 
}\label{lowC:prepend}
\For{$t \in N'$\label{lowC:onlinestart}}{
    Agent $t$ arrives and reveals type $j$.\\
    Pick a bag from $G_j$ and assign it to agent $t$. Remove the bag from $G_j$. 
}\label{lowC:onlineend}
\end{algorithm}

\subsubsection{Notations Used for the Analysis of \Cref{alg:lowC}}\label{sub:knownDnotations}
For an input instance \(\I' = (N', M', \{v_j\}_{j=1}^k)\), where \(n' = |N'|\), given a distribution \(D\) where the probability of an agent being of type \(j\) is \(p_j\) (for any \(j \in [k]\)), \(\mu_j = n' p_j\) is defined as the expected number of agents of type \(j\). We introduce a parameter \(z\) defined as
\begin{equation}\label{definez}
    z = n'^\epsilon \sum_{j=1}^{k} \sqrt{\mu_j}
\end{equation}
which serves as an approximate upper bound on the total number of extra bundles saved ex-post.

We define a function \(b(i)\) that maps type \(i\) to the set of types preceding it in the ordering \(L\). 

For each type \(j\), let \(E_j\) denote the set of singleton bundles reserved in \(G_j\) during steps \ref{turnstart}--\ref{turnend}. Since each item is reserved for at most one type, the reserved sets for distinct types are disjoint, and since \(|G_j|\) is upper bounded by \(M_j\), we have 
\begin{equation}\label{eq:upperE_j}
    |E_j| \le |G_j| \le M_j.
\end{equation}

Let us partition \(T_i\) as follows
\begin{equation}\label{eq:T_ibar}
    \Bar{T_i} = T_i \cap \bigcup_{j \in b(i)} E_j \quad \text{and} \quad T'_i = T_i \setminus \Bar{T_i}
\end{equation}
where, \(\Bar{T_i}\) represents the set of items that are highly valued by type \(i\) but have already been assigned to types preceding \(i\) in the ordering, and \(T'_i\) represents the set of high-valued items available for type \(i\) to choose when its turn begins.

\subsubsection{Analysis of \Cref{alg:lowC}}

Since \cref{prop:limitT_i} and \cref{lem:oneunsat} will be used in the analysis of \cref{sec:unknowndist}, we restate the full notation for clarity. In \cref{prop:limitT_i}, we provide an upper bound on the number of high-valued items for every unsaturated type. Then, we use this bound in \cref{lem:allsaturated} to prove that all types are saturated in the preprocessing phase. Given this, we finally establish that, with high probability, the reserved bundles for each type are sufficient to accommodate all arriving agents of that type.

\begin{proposition}\label{prop:limitT_i}
Given an input instance $\I' = (N', M', \{v_i\}_{i\in [k]})$ with known distribution $D$, where $n' = |N'|$, if by the end of line \ref{turnend} in \Cref{alg:lowC} type \(i\) is not saturated, then 
\[
|T_i| \le n'\Bigl(1-\frac{p_k}{2}\Bigr) + z,
\]
where \(z = n'^\epsilon \sum_{j=1}^{k} \sqrt{n'p_j}\), with \(p_j\) denoting the probability that an agent is of type \(j\) under \(D\) (in particular, \(p_k\) is the probability of the least frequent type).
\end{proposition}

\begin{proof}
Since type $i$ is not saturated by the end of line \ref{turnend}, in the turn of type $i$, the while loop in line \ref{whilehighval} only stops because all available high-valued items for type \(i\) (i.e., \(T'_i\)) are reserved as singleton bundles. Therefore $|E_i|= |T'_i| < M_i$. We consider three cases:

\begin{itemize}[leftmargin=*]
    \item \textbf{Case 1: \(i \neq L[k]\).} As \(\Bar{T_i} = T_i \cap \bigcup_{j\in b(i)}E_j\subseteq \cup_{j\in b(i)}E_j\), using the upper bound in \eqref{eq:upperE_j}, we have that
    \(|\Bar{T_i}| \leq \sum_{j\in b(i)} M_j\). Since \(|T_i| = |T'_i| + |\Bar{T_i}| \le M_i + \sum_{j\in b(i)} M_j\) and each \(M_j \le n'p_j + n'^\epsilon\sqrt{n'p_j}\), it follows that
    \begin{align*}   
    |T_i| \le & n'\sum_{j\in b(i)\cup\{i\}}p_j + z \le \\ &n'(1-p_{L[k]}) + z \le n'(1-p_k) + z < n'(1-\frac{p_k}{2}) + z,
    \end{align*}

    where we used the fact that \(L(k) \notin b(i)\cup\{i\}\) and that the probability of any type is at least \(p_k\).
    
    \item \textbf{Case 2 (\(i = L[k]\) and \(|T_{L[1]}\setminus T_i| \ge \frac{n'p_{L[1]}}{2}\)):} Let \(u = L[1]\) be the first type in the ordering of $L$. By line \ref{prioritizing}, type \(u\) reserves at least \(\frac{n'p_u}{2}\) items outside \(T_i\), so at most \(M_u - \frac{n'p_u}{2}\) of its reserved items belong to \(\Bar{T_i}\). For any other \(j \in b(i)\setminus u\), we use the general bound \(|E_j \cap \Bar{T_i}| \le |E_j| \leq M_j\), which implies $|\Bar{T_i}| \le \sum_{j \in b(i)} M_j - \frac{n'p_u}{2}$. Since \(|T_i| = |T'_i| + |\Bar{T_i}|\) with \(|T'_i| < M_i\), it follows that
\begin{equation}\label{eq:boundT_i}
    \begin{aligned}
        |T_i| < & \sum_{j \in b(i) \cup \{i\}} M_j - \frac{n'p_u}{2} = \\
        & \sum_{j \in [k]} M_j - \frac{n'p_u}{2} \leq n' +z - \frac{n'p_u}{2} \leq n'\Bigl(1-\frac{p_k}{2}\Bigr) + z
    \end{aligned}
\end{equation}
where we used \(b(i) \cup \{i\} = [k]\) as \(i\) is the last type in the ordering \(L\), $\sum_{i\in[k]} M_j \leq n'+z$, and the fact that $p_k$ is the smallest probability, i.e. $p_u\geq p_k$. 

\item \textbf{Case 3 (\(i = L[k]\) and \(|T_{L[1]}\setminus T_i| < \frac{n'p_{L[1]}}{2}\)):} In this case, the condition in line \ref{sortstart} is false, so for every pair of types \(u, v\) we have \(|T_u \setminus T_v| < \frac{n'p_u}{2}\). In this scenario, \(i = L[k]\) is the least frequent type (i.e. \(i = k\)). Suppose for contradiction that \(|T_k| > n'\Bigl(1-\frac{p_k}{2}\Bigr) + z\). For any \(j \in [k-1]\) we have that \(|T_k\setminus T_j| < \frac{n'p_k}{2}\), implying \(|T_k\cap T_j| \ge |T_k| - \frac{n'p_k}{2}\). In particular, for \(j = 1\) we obtain
\[
|T_k \cap T_1| \ge n'\Bigl(1-\frac{p_k}{2}\Bigr) + z - \frac{n'p_k}{2}.
\]
Since \(T_k \cap T_1 \subseteq T_k\), and \(|T_k\setminus T_2| < \frac{n'p_k}{2}\), at most \(\frac{n'p_k}{2}\) items in $T_k \cap T_1$ are not high-valued for type 2. Hence,
\[
|T_k \cap T_1 \cap T_2| \ge n'\Bigl(1-\frac{p_k}{2}\Bigr) + z - 2\frac{n'p_k}{2}.
\]
Continuing this argument for every \(j < k\) leads to
\[
\Bigl|\bigcap_{j=1}^{k-1} T_j \cap T_k\Bigr| \ge n' + z - k\frac{n'p_k}{2} \geq \frac{n'}{2} + z,
\]
where we used \(p_k \le \frac{1}{k}\). This contradicts the fact that no item is high-valued by all \(k\) types in the reduced instance. Therefore, we must have \(|T_k| \le n'\Bigl(1-\frac{p_k}{2}\Bigr) + z\).
\end{itemize}
\end{proof}

\noindent The following technical proposition will be needed in the proof of \Cref{lem:allsaturated}.

\begin{restatable}{proposition}{propboundk}\label{prop:boundk}
The following bounds hold. 
\begin{enumerate}
    \item If \(p_k = \omega\Bigl(\frac{n^\epsilon k}{\sqrt{n}}\Bigr)\), then \(k = o\Bigl(n^{\frac{1}{4} - \frac{\epsilon}{2}}\Bigr)\).
    \item If in \cref{alg:knownD}, $|C| < n(1-\frac{1}{k}) + n^{\epsilon}\sqrt{np_1}$, number of agents in $\I'$ satisfies \(n' = \Omega\Bigl(\frac{n}{k}\Bigr) = \omega\Bigl(n^{\frac{\epsilon}{2}+\frac{3}{4}}\Bigr)\). 
\end{enumerate}
\end{restatable}
\begin{proof}
    \begin{enumerate}
        \item As $p_k$ is the probability of least frequent type, \(p_k \le \frac{1}{k}\). Under the assumption \(p_k = \omega\left(\frac{n^\epsilon k}{\sqrt{n}}\right)\), it follows that \(\frac{n^\epsilon k}{\sqrt{n}} = o(p_k) \le o\Bigl(\frac{1}{k}\Bigr)\). Multiplying by \(k\) yields \(\frac{n^\epsilon k^2}{\sqrt{n}} = o(1)\), i.e., \(k^2 = o\Bigl(n^{\frac{1}{2}-\epsilon}\Bigr)\). Taking square roots on both sides gives \(k = o\Bigl(n^{\frac{1}{4}-\frac{\epsilon}{2}}\Bigr)\).\\

        \item Note that \(n' = n-|C| \ge \frac{n}{k} - n^\epsilon\sqrt{np_1} \ge \frac{n}{k} - n^{\epsilon+\frac{1}{2}}\). As \(k = o\Bigl(n^{\frac{1}{4}-\frac{\epsilon}{2}}\Bigr)\), so \(\frac{n}{k} = \omega\Bigl(n^{\frac{\epsilon}{2}+\frac{3}{4}}\Bigr)\) and hence \(n' = \Omega\Bigl(\frac{n}{k}\Bigr)\) for \(\epsilon < \frac{1}{2}\).
    \end{enumerate}
\end{proof}

\begin{restatable}{lemma}{lemoneunsat}\label{lem:oneunsat}
    Given an input instance $\I' = (N', M', \{v_i\}^k_{i=1})$ with known distribution $D$, where $n' = |N'|$, if by the end of preprocessing steps (line \ref{lowC:prepend}) of \cref{alg:lowC}, a type $i$ is not saturated, $|G_i| \geq |T_i| - n' + \frac{ v_i(M')-|T_i|}{2\alpha} + n'p_i - n'^\epsilon \sum_{j\neq i}\sqrt{n'p_j}$, where, for any \(j \in [k]\), \(p_j\) denotes the probability of an agent being of type \(j\) according to the distribution \(D\).
\end{restatable}
\begin{proof}
    Since type \(i\) is not saturated by the end of preprocessing (i.e., \(|G_i| < M_i\)), and reserved bundles are only added to \(G_i\) during this phase, the total number of reserved bundles is always below the threshold \(M_i\). Consequently, the while loop in line \ref{whilehighval} terminates once all available items in \(T'_i\) have been added as singleton bundles to \(G_i\). Therefore,
    \begin{equation}\label{eq:E_i}
        E_i = T'_i.
    \end{equation}
    
    In line \ref{BFstart}, \(R\) denotes the set of items remaining after all types have arrived in order and selected \(|E_j|\) available items from their respective high-valued sets. Hence $R = M' \setminus \bigcup^k_{j=1} E_j$. Note that $T_i \subseteq \bigcup^k_{j=1} E_j$ because we splitted $T_i$ so that $T_i = T'_i \cup \Bar{T}_i$, and $T'_i = E_i$ by \eqref{eq:E_i} and $\Bar{T}_i \subseteq \bigcup_{j\in b(i)} E_j$ by \eqref{eq:T_ibar}. Therefore $R\cap T_i= \emptyset$, hence 
    \begin{equation}\label{eq:Rlowval}
        \forall g\in R: v_i(g)< \alpha.
    \end{equation}
    All items in $T_i$ are high-valued by type $i$, but since each item's value is bounded by $1$ in a normalized instance, $v_i(T_i) \leq |T_i|$. Every item in $\bigcup^k_{j=1} E_j \setminus T_i$ is low valued by type $i$, and hence can have a value of at most $\alpha$. Note that $T'_i = E_i$ implies $\bigcup^k_{j=1} E_j \setminus T_i = \bigcup_{j\in[k],j\neq i} E_j \setminus \Bar{T}_i$. For ease of notation, from now on we only use $j\neq i$ to denote $\{j\in [k] \mid j\neq i\}$. Putting it together, we obtain
    \begin{align*}
        v_i(R) = & v_i(M') - v_i(T_i) - v_i(\bigcup^k_{j=1}E_j\setminus T_i) \geq \\& v_i(M')-|T_i| - (\sum_{j\neq i} |E_j| - |\Bar{T_i}|)\alpha. \end{align*}

Let us refer to lines \ref{while:BF}–\ref{lowC:prepend} as the bag-filling process. During this process, each type \(j\) can obtain up to \(M_j - |E_j|\) extra bundles since \(|G_j| \le M_j\) and \(E_j\) contains the bundles already reserved in an earlier step (i.e., \(|G_j \setminus E_j| \le M_j - |E_j|\) for all \(j\)). Hence, type \(i\) can lose at most \(\sum_{j\neq i} \left(M_j - |E_j|\right)\) bags to other types during the bag-filling process. Recall that a type only claims a bag if its value exceeds \(\alpha\). Moreover, immediately before the last item is added to each assigned bag in the bag-filling process, the bag remains unclaimed by all unsaturated types (including type \(i\)), and by \eqref{eq:Rlowval}, every remaining item in \(R\) is valued below \(\alpha\) by type \(i\). Therefore, when such an item is added as the last item to a bag, it can increase the bag's value to at most \(2\alpha\). Consequently, during bag-filling the value of every assigned bundle is at most \(2\alpha\) for type \(i\). It follows that the remaining value for type \(i\) is at least 
\[
v_i\Bigl(R\setminus \bigcup_{j\neq i} G_j\Bigr) \ge v_i(M') - |T_i| - \Bigl(\sum_{j\neq i}|E_j| - |\Bar{T_i}|\Bigr)\alpha - \sum_{j\neq i}\Bigl[M_j - |E_j|\Bigr]2\alpha.
\]
Since the value of every assigned bundle is at most \(2\alpha\) (including those reserved for type \(i\)), the number of bags reserved for type \(i\) during the bag-filling is at least 
\[
\frac{v_i\Bigl(R\setminus \bigcup_{j\neq i} G_j\Bigr)}{2\alpha}.
\]
Putting this together with the number of bags reserved for type $i$ in the earlier step in $E_i$ and \ref{eq:E_i}, and considering $M_i=\lfloor\mu_i + n'^\epsilon\sqrt{\mu_i}\rfloor$, we obtain the following:

\begin{equation}\label{eq:firstboundG_i}
\begin{aligned}
|G_i| \geq & |T'_i| + \frac{ v_i(M')-|T_i|}{2\alpha} - \frac{(\sum_{j\neq i} |E_j - |\Bar{T_i}|)}{2} \\& -\sum_{j\neq i} \left[n'p_j + n'^\epsilon \sqrt{n'p_j} - |E_j|\right] \\
& = |T'_i| + \frac{ v_i(M')-|T_i|}{2\alpha} - \frac{\sum_{j\neq i} |E_j|}{2} + \frac{|\Bar{T_i}|}{2} \\& - n'(1-p_i) - n'^\epsilon \sum_{j\neq i}\sqrt{n'p_j} + \sum_{j\neq i} |E_j| \\
& = |T'_i| + \frac{ v_i(M')-|T_i|}{2\alpha} + \frac{\sum_{j\neq i} |E_j|}{2} + \frac{|\Bar{T_i}|}{2}\\& - n' + n'p_i - n'^\epsilon \sum_{j\neq i}\sqrt{n'p_j}
\end{aligned}
\end{equation}

where we used $\sum_{j\neq i} n'p_j = n'\sum_{j\neq i} p_j = n'(1-p_i)$. By \eqref{eq:T_ibar}, $\Bar{T}_i \subseteq \bigcup_{j\in b(i)} E_j$, hence $|\Bar{T_i}|\leq \sum_{j\neq i} |E_j|$. Using this inequality in \eqref{eq:firstboundG_i} we derive
\begin{equation}\label{eq:Gi1}
    \begin{aligned}
    |G_i| & \geq |T'_i| + |\Bar{T_i}| + \frac{ v_i(M')-|T_i|}{2\alpha} -n' + n'p_i - n'^\epsilon \sum_{j\neq i}\sqrt{n'p_j}\\
    & = |T_i| - n' + \frac{ v_i(M')-|T_i|}{2\alpha} + n'p_i - n'^\epsilon \sum_{j\neq i}\sqrt{n'p_j}
\end{aligned}
\end{equation}
where we used $T_i = T'_i\cup \Bar{T}_i$.

\end{proof}

\noindent The following lemma is our main guarantee for \Cref{alg:lowC}.
\begin{lemma}\label{lem:allsaturated}
    If \cref{alg:knownD}, invokes \cref{alg:lowC} on an instance $\I'=(N',M',\{v_i\}^k_{i=1})$, where for all $i\in k$, $v_i(M')\geq |N'|$ and $p_{k} = \omega(\frac{n^\epsilon k}{\sqrt{n}})$, for any small $\eta>0$ where $\alpha = \frac{1}{2(1+\eta)}$, there exists an $n(\eta)$, where for all $n>n(\eta)$, all types are saturated during the preprocessing steps(lines \ref{lowC:prepstart}-\ref{lowC:prepend}) of \cref{alg:lowC}. 
\end{lemma}
\begin{proof}

As $n'= |N'|$, \(v_i(M') \ge n'\) for every \(i \in [k]\). Now suppose for contradiction that type $i$ is not saturated in \cref{alg:lowC}. According to \cref{lem:oneunsat}, 
    \begin{equation}\label{eq:Gi}
    \begin{aligned}
    |G_i| & \geq|T_i| - n' + \frac{ v_i(M')-|T_i|}{2\alpha} + n'p_i - n'^\epsilon \sum_{j\neq i}\sqrt{\mu_j} \\
    & \geq \left[\frac{1}{2\alpha}-1\right] (n'-|T_i|) + n'p_i - n'^\epsilon \sum_{j\neq i}\sqrt{\mu_j}
\end{aligned}
\end{equation}
where we used $\mu_j = n'p_j$, and $v_i(M')\geq n'$. Proposition \ref{prop:limitT_i} implies $|T_i| \leq n'(1-\frac{p_k}{2}) + z$. Hence, we obtain: 
\begin{align}\label{eq:onetolast}
    |G_i| & \geq \left[\frac{1}{2\alpha}-1\right]  \left[\frac{n'p_k}{2} - z\right]+ n'p_i - n'^\epsilon \sum_{j\neq i}\sqrt{\mu_j}.
\end{align}

By the Cauchy-Schwarz inequality, we have
\[
z = n'^\epsilon \sum_{i=1}^{k}\sqrt{\mu_i} \le n'^\epsilon \sqrt{k \sum_{i=1}^{k}\mu_i} = n'^\epsilon \sqrt{k n'}.
\]
By assumption, \(p_k = \omega\Bigl(\frac{n^\epsilon k}{\sqrt{n}}\Bigr)\). \Cref{prop:boundk} implies \(n' = \Omega\Bigl(\frac{n}{k}\Bigr)\). 
Therefore, $p_k = \omega\Bigl(\frac{n^\epsilon\sqrt{k}}{\sqrt{n'}}\Bigr)$. Since \(n \ge n'\), it follows that $p_k = \omega\Bigl(\frac{n'^\epsilon\sqrt{k}}{\sqrt{n'}}\Bigr)$. Hence, it follows that $p_k = \omega\Bigl(\frac{z}{n'}\Bigr)$. Consider $\alpha =\frac{1}{2(1+\eta)}$ given any small $\eta$. There exists a $n(\eta)$ where for all $n> n(\eta)$, $\eta (\omega(z)-z) \geq z$. Let us consider this range of $n$, then we can rewrite \eqref{eq:onetolast} as follows: 

\begin{align}\label{eq:last}
    |G_i| & \geq \eta \left[\omega(z) -z\right]+ n'p_i - n'^\epsilon \sum_{j\neq i}\sqrt{\mu_j} \\
    & \geq z + n'p_i - n'^\epsilon \sum_{j\neq i}\sqrt{\mu_j} \\
    & = n'p_i+ n'^\epsilon\sqrt{\mu_i} \geq M_i
\end{align}

However, this is a contradiction with the assumption that type $i$ is unsaturated, i.e. $|G_i| < M_i$. Therefore, all types should be saturated. 
\end{proof}

Finally, we are ready to prove our main lemma for this Section. 
\lemsecondimp*
\begin{proof}
The first \( |C| \) agents each receive a singleton bundle from \( C \), and since every item in \( C \) is valued above \(\alpha\) by all types, their MMS requirements are met. For the reduced instance $\Ical' = \bigl([|C|+1,n],\, M',\, \{v_i\}_{i=1}^k\bigr)$ with \( n' = n - |C| \). Let \( X^{n'}_i \) denote the number of type \( i \) agents in \(\Ical'\). By Proposition~\ref{prop:boundx_i} and the union bound,
\[
P\Bigl[\exists\, i\in[k]: X^{n'}_i > M_i\Bigr] \le O\Bigl(k\,e^{-n'^{2\epsilon}}\Bigr).
\]
Since Proposition~\ref{prop:boundk} gives \( k = o\bigl(n^{\frac{1}{4}-\frac{\epsilon}{2}}\bigr) \) and \( n' = \omega\bigl(n^{\frac{\epsilon}{2}+\frac{3}{4}}\bigr) \), the bound simplifies to \( O\bigl(e^{-n'^{2\epsilon}}\bigr) = o\bigl(e^{-n^{1.5\epsilon}}\bigr) \). Thus, with probability at least \( 1 - o\bigl(e^{-n^{1.5\epsilon}}\bigr) \), we have \( X^{n'}_i \le M_i \) for all \( i \).

Moreover, since each preallocated item is valued at most \(1\), it holds that \( v_i(M') \ge n - |C| = n' \) for all \( i \). Then, by Lemma~\ref{lem:allsaturated}, there exists an \( n(\eta) \) such that for all \( n > n(\eta) \) every type becomes saturated during preprocessing (i.e., \( |G_i| = M_i \) for all \( i \)). This ensures that the number of reserved bundles for each type is sufficient to cover all type-\(i\) agents in \(\Ical'\), and hence every agent receives a bundle valued above \(\alpha\) with high probability.
\end{proof}

\subsection{Tightness of Analysis}\label{app:tightness}
The following claims show that our analysis is almost-tight.

\begin{restatable}{claim}{claimtightone}
    Algorithm \ref{alg:knownD} cannot obtain an approximation factor better than $\alpha = \frac{1}{2}$.
\end{restatable}
\begin{proof}
We demonstrate via an example that an \(\alpha\)-competitive guarantee of \(\alpha = \frac{1}{2} + \epsilon\) is unattainable for any \(\epsilon > 0\). Consider \(k\) types with equal arrival probabilities and \(2n\) items. The first type values every item at \(\frac{1}{2}\), while each of the remaining types values all but the last two items at \(\frac{1}{2}\). For the \(i\)th type (with \(2\leq i \leq k\)), the last two items are valued at \(\frac{1}{2} + \frac{\epsilon}{2i}\) and \(\frac{1}{2} - \frac{\epsilon}{2i}\), respectively. Note that the MMS value of all types is $1$. Since no type possesses a highly valued item, reserved bundles in the preprocessing step are allocated solely through the bag-filling procedure. In this process, each bag consumes exactly two items and is arbitrarily assigned to one of the types, resulting in at most \(n\) bags in total. However, to saturate all types, each type must receive more than \(\frac{n}{k}\) bundles, necessitating more than \(n\) bags overall. This contradiction shows that it is impossible to saturate all types when \(\alpha > \frac{1}{2}\).
\end{proof}

\begin{restatable}{claim}{claimtighttwo}
    Algorithm~\ref{alg:knownD} fails to be \(\alpha\)-MMS competitive for any \(\alpha > 0\) if \(p_k = o\Bigl(\frac{n^\epsilon k}{\sqrt{n}}\Bigr)\).
\end{restatable}
\begin{proof} 
    We demonstrate via an example that, to achieve an \(\alpha\)-MMS guarantee for any \(\alpha > 0\), the minimum type probability must satisfy $
    p_k = \Omega\Bigl(\frac{n^\epsilon k}{\sqrt{n}}\Bigr)$.
    
    Consider an instance with \(k\) types where each type values exactly \(n\) items at 1. Assume that for \(i\in [k-2]\) we have \(p_i = \frac{1}{k-1}\), \(p_{k-1} = \frac{1}{k-1} - p_k\). Suppose there are \(m = n + \frac{n}{k}\) items available, with only the first \(n - \frac{n}{k}\) items valued at 1 by all types; thus, \(|C| = n\Bigl(1 - \frac{1}{k}\Bigr)\). \Cref{alg:knownD} assigns all items in $C$ to the first $|C|$ agents, then, \cref{alg:lowC} is invoked on the reduced instance \(\I' = (N', M', \{v_i\}_{i=1}^k)\) with \(n' = n - |C| = \frac{n}{k}\) and \(M' = M \setminus C\), so that \(M'\) contains \(2n'\) items (assume \(M' = \{1,2,\dots,2n'\}\)). 

In the reduced instance, each type \(i\in [k-1]\) values the first \(n'-1\) items and the \((n'-1+i)\)th item at 1, while type \(k\) values the last \(n'\) items at 1. Consequently, for any \(i\in [k-1]\) we have 
\[
|T_k\setminus T_i| = n'-1 \ge \frac{n'p_k}{2},
\]
and for any \(j\in [k-1]\), 
\[
|T_i\setminus T_j| = 1 < \frac{n'p_i}{2}.
\]
Thus, type \(k\) is selected as the first type in the ordering. Each of the first \(k-1\) types values exactly one item from \([n',2n']\) at 1; if such an item is not picked by type \(k\), it can be selected by the corresponding type, but this single item does not affect the asymptotic bound.

To ensure that the first \(n'-1\) items are divided among the first \(k-1\) types so that each type \(i\) receives 
\[
n'p_i + \Theta\Bigl(n'^\epsilon\sqrt{\mu_i}\Bigr)
\]
bundles, we require
\[
\sum_{i\in [k-1]} \Bigl( n'p_i + \Theta\Bigl(n'^\epsilon\sqrt{\mu_i}\Bigr) \Bigr) \le n'-1.
\]
This implies 
\[
\Theta\Bigl(k\cdot n'^\epsilon\sqrt{\frac{n'}{k}}\Bigr) \le n'p_k,
\]
or equivalently,
\[
p_k = \Omega\Bigl(\frac{n'^\epsilon \sqrt{k}}{\sqrt{n'}}\Bigr).
\]
Since \(n' = \frac{n}{k}\), this further implies
\[
p_k = \Omega\Bigl(\frac{n'^\epsilon k}{\sqrt{n}}\Bigr).
\]
For constant \(k\) (with \(n' = \Theta(n)\)), the general requirement becomes
\[
p_k = \Omega\Bigl(\frac{n^\epsilon k}{\sqrt{n}}\Bigr).
\]

\end{proof}

\section{Stochastic Arrival of Agents with Unknown Distribution}\label{sec:unknowndist}
In this section, we study the \textsc{OnlineKTypeFD} problem under the stochastic arrival model, where agents arrive according to an \emph{unknown} probability distribution $D$. 
We present an algorithm that achieves an \(\alpha\)-MMS competitive guarantee under mild assumptions, with \(\alpha\) approaching \(\frac{1}{2}\) as \(n\) grows large.

The rest of this section is organized as follows. In \Cref{subsec:unknownD} we present the main algorithm (\Cref{alg:unknownD}) that handles all instances. The rationale behind the choices of the parameters \(\epsilon\) and \(\epsilon'\) used in \cref{alg:unknownD}, and the auxiliary lemmas used for the analysis of this algorithm are provided in \cref{app:notations} and \cref{app:auxiliary} respectively. Based on the number of universally high-valued items, \cref{alg:unknownD} either invokes \cref{alg:knownD}, or invokes Algorithms \ref{alg:adversarial-alg} and \ref{alg:lowC}. The analysis of \cref{alg:unknownD} when each of these algorithms are invoked is provided in Sections~\ref{app:invoke5}, \ref{app:invoke2}, and \ref{app:invokeend} respectively.

\subsection{An (Almost) $\frac{1}{2}$-MMS-competitive Algorithm}\label{subsec:unknownD}
\Cref{alg:unknownD} is our main algorithm for this model. Given the input parameter \(\alpha\), \cref{alg:unknownD} first identifies the set \(C\) of universally high-valued items. This algorithm operates in two cases depending on the number of universally high-valued items. In the first case, when the number of such items \(C\) is at least \(\lceil n^{\epsilon'} \rceil\), the algorithm allocates one item from \(C\) to each of the first \(\lceil n^{\epsilon'} \rceil\) agents, thereby ensuring that these agents receive a bundle they value at least $\alpha$ while simultaneously gathering sufficient data to estimate the underlying type distribution via \cref{alg:learndist}. Using the estimated distribution \(\hat{D}\), the problem is then reduced to a stochastic arrival model with known distribution, and the remaining items are allocated by \cref{alg:knownD}. In the second case, when \(|C| < \lceil n^{\epsilon'} \rceil\), all items in \(C\) are allocated to the first \(|C|\) agents, and the remaining items are partitioned into two sets, \( B_1 \) and \( B_2 \), where \( B_1 \) is designated for the next \( \lceil n^{\epsilon'} \rceil \) agents, and \( B_2 \) for the rest of the agents. Each item is independently assigned to \( B_1 \) with probability \( p \), and to \( B_2 \) otherwise. The probability \(p\) is chosen sufficiently high to ensure that, regardless of the underlying distribution, we can assign each of the next \(\lceil n^{\epsilon'} \rceil\) agents a unique bundle—formed from items in \(B_1\)—that is valued at more than \(\alpha\) by that agent. We then apply the algorithm for adversarial arrival (\cref{alg:adversarial-alg}) to allocate items in \( B_1 \) to the next \( \lceil n^{\epsilon'} \rceil \) agents, and by observing the types of these agents, we learn an estimate of the distribution by invoking \cref{alg:learndist}. Finally, using the estimated distribution, we invoke \cref{alg:lowC} to allocate items in \( B_2 \) to the remaining agents.

The main result of this section is stated in \cref{thm:unknownD}, whose proof is built on three key lemmas: lemmas \ref{lem:alg5works}, \ref{lem:alg2works}, and \ref{lem:algendworks}. Detailed proofs of these lemmas are provided in Sections~\ref{app:invoke5}, \ref{app:invoke2}, and \ref{app:invokeend}, respectively.

\begin{restatable}{lemma}{lemalgfiveworks}\label{lem:alg5works}
    If \cref{alg:unknownD} invokes \cref{alg:knownD}
    on an instance $\Icalhat$, where $p_k = \omega(\frac{k}{n^{\frac{2}{9}(1-c)}})$, for any small $\eta>0$, there exists an $n(\eta)$, where for all $n>n(\eta)$, \cref{alg:knownD} is \( \frac{1}{2(1+\eta)}-\)MMS competitive with probability at least \( 1 - o(e^{-n^{c/2}}) \).
\end{restatable}

\begin{restatable}{lemma}{lemalgtwoworks}\label{lem:alg2works}
    If \cref{alg:unknownD} invokes \cref{alg:adversarial-alg} on an instance $\Icalhat$, where $n>10$, for any $\alpha\leq \frac{1}{2}$, \cref{alg:adversarial-alg} is \( \alpha-\)MMS competitive with probability at least \( 1 - O(e^{-n^c}) \).
\end{restatable}

\begin{restatable}{lemma}{lemalgendworks}\label{lem:algendworks}
    If \cref{alg:unknownD} invokes \cref{alg:lowC} on an instance $\I'$, where $p_k = \omega(\frac{k}{n^{\frac{2}{9}(1-c)}})$, for any small $\eta>0$, there exists an $n(\eta)$, where for all $n>n(\eta)$, \cref{alg:lowC} is \( \frac{1}{2(1+\eta)}-\)MMS competitive with probability at least \( 1 - o(e^{-n^{c/2}}) \).
\end{restatable}

Assuming the validity of the three aforementioned lemmas, we now proceed to prove our main theorem.

\thmunknownD*
\begin{proof}
Agents receiving singleton bundles with universally high-valued items meet their MMS requirement. Moreover, lemmas~\ref{lem:alg5works}, \ref{lem:alg2works}, and \ref{lem:algendworks} collectively guarantee that all remaining agents also obtain bundles meeting their MMS guarantee.
\end{proof}

\begin{algorithm}[h!]
\caption{Learn Distribution}\label{alg:learndist}
    \SetKwInOut{Input}{Input}\SetKwInOut{Output}{Output}
  \Input{Types of last $\lceil n^{\epsilon'} \rceil$ agents arrived}
  \Output{Estimation of distribution $\hat{D}$.}
  \BlankLine
    For each type $j\in[k]$, let $X_j$ be the number of type $j$ agents among the last $\lceil n^{\epsilon'} \rceil$ agents.\\
    Let $\hat{D}$ be the estimated distribution where the arrival probability of type $j$ is estimated as $\hat{p_j} = \frac{X_j}{\lceil n^{\epsilon'} \rceil}.$\\
    Return $\hat{D}$
\end{algorithm}

\begin{algorithm}[h!]
\caption{Stochastic Arrivals of Agents with Unknown Distribution}\label{alg:unknownD}
    \SetKwInOut{Input}{Input}\SetKwInOut{Output}{Output}
  \Input{A Stochastic Instance $\I = (N,M, \{v_i\}^k_{i=1})$, $\alpha = \frac{1}{2(1+\eta)}$ for a small $\eta$, a small constant $0<c<0.1$.}
  \Output{ $\alpha-$MMS.}
  \BlankLine
Let $n = |N|$, $\epsilon = \frac{5+4c}{18}$, $\epsilon' = \frac{2+c}{3}$. \label{algend:define}\\
Let $C = \{g\in M \mid \forall i\in [k] :v_i(g)\geq \alpha \}$, $M' = M\setminus C$.\label{algend:preproc:C}\\
\eIf{$|C|\geq \lceil n^{\epsilon'} \rceil$\label{algend:firststart}}{
Give one item from $C$ to each of the first $\lceil n^{\epsilon'} \rceil$ agents, and let the remained items of $C$ be $C^0$. \\
Let $\hat{D}$ be the estimated distribution obtained from \cref{alg:learndist}.\\
Let $\Icalhat = ([\lceil n^{\epsilon'} \rceil+1, n], M'\cup C^0,\{v_j\}^k_{j=1})$.\label{algend:Icalhat}\\
Run \cref{alg:knownD} on $(\Icalhat,\hat{D},\alpha,\epsilon)$. \label{algend:firstend}
}{\label{algend:secondstart}
Let $C^0 = C$.\\
Give one item from $C$ to each of the first $|C|$ agents. \\
Let $B_1 = \emptyset$ and $B_2 = \emptyset$, $p = \frac{2k\lceil n^{\epsilon'} \rceil}{n-|C^0|}$. \label{alg6:firstCend}\\
Add each item of $M'$ to $B_1$ with probability $p$, otherwise add it to $B_2$. \label{alg6:constructB_1}\\
Run \cref{alg:adversarial-alg} on $\Icalhat = ([|C^0|+1, |C^0| + \lceil n^{\epsilon'} \rceil], B_1, \{v_j\}^k_{j=1}).$\\
Let $\hat{D}$ be the estimated distribution obtained from \cref{alg:learndist}.\\
Run \cref{alg:lowC} on $\left(\Ical' = ([|C^0|+\lceil n^{\epsilon'} \rceil + 1,n], B_2, \{v_j\}^k_{j=1}),\hat{D}, \epsilon, \alpha\right)$. \label{algend:secondend}
}
\end{algorithm}

\subsection{Notations and Parameter Selection}\label{app:notations}
We define the parameter \begin{equation}\label{dfn:delta}
    \delta = \frac{1}{n^{\frac{\epsilon' - c}{2}}}.
\end{equation}
which will be used in \cref{lem:B1} to provide a lower bound for the value of the bags \(B_1\) and \(B_2\) constructed by \cref{alg:unknownD}. Later, we will show that with high probability, for every \(i \in [k]\) and every \(B_j \in \{B_1, B_2\}\), $v_i(B_j) \geq (1-\delta)\mathbb{E}(v_i(B_j))$. \\
 
Note that the model in this section cannot exceed the bound obtained in the previous section. By the tightness of the analysis in Theorem~\ref{thm:knownD}, the best achievable bound is \( p_k = \omega\left(\frac{n^{\epsilon} k}{\sqrt{n}}\right) \), which implies  \( k = o(n^{\frac{1}{4} - \frac{\epsilon}{2}}) \) by \cref{prop:boundk}. To make these bounds as general as possible, we aim to minimize \(\epsilon\). 

Parameter \( c \) can be chosen as an arbitrarily small positive constant, so let us consider any \( 0 < c < 0.1 \). Consequently, the inequalities in \cref{cl:alleqs}—which are essential for our subsequent analysis—impose a lower bound of \(\frac{5+4c}{18}\) on \(\epsilon\). Hence, to minimize $\epsilon$, we select parameter $$\epsilon = \frac{5+4c}{18}.$$ 
This contrasts with the previous section, where \(\epsilon\) could be arbitrarily small.

To use the framework established in the earlier section, we need to ensure $\epsilon<\frac{1}{2}$, which is satisfied for any $c<0.1$. Consequently, we obtain the bounds:
\begin{equation}\label{boundk}
    k = o\left(n^{\frac{1 - c}{9}}\right), \quad p_k = \omega\left(\frac{k}{n^{\frac{2}{9}(1 - c)}}\right).
\end{equation}

\begin{claim}\label{cl:alleqs} The smallest $\epsilon$ that satisfies the following set of inequalities is $\epsilon = \frac{5+4c}{18}$ when $\epsilon' = \frac{2+c}{3}$. 
     \begin{enumerate}
        \item  $\epsilon' \leq 1.5\epsilon + \frac{1}{4}$.
        \item $1 \leq \frac{\epsilon'}{2} - \frac{c}{2} + 1.5\epsilon + \frac{1}{4}$.
     \end{enumerate}
\end{claim}

Considering our chosen \(\epsilon\), \cref{cl:alleqs} implies that we must select the parameter 
\[\epsilon' = \frac{2+c}{3}.\] 
We will later show that observing only \(\lceil n^{\epsilon'} \rceil\) agents (and their types) is sufficient to learn the distribution $D$.

\subsection{Technical Lemmas}\label{app:auxiliary}
In this section, we will provide some important propositions and lemmas that will be used for proving the main lemmas of this section. 

\begin{proposition}\label{lem:probpredict}
    For a type $i\in[k]$, and any positive $\gamma$, $P(|p_i - \hat{p_i}| > \gamma) < 2e^{-\gamma n^{\frac{\epsilon'}{2}}} $ .
\end{proposition}
\begin{proof}
    For type $i$, $\hat{p_i}$ is the sample mean of random variables drawn from a Bernoulli distribution with mean $p_i$. 
    By CLT, we have that $\hat{p}_i \sim \mathcal{N} \left( p_i, \frac{p_i(1-p_i)}{\lceil n^{\epsilon'} \rceil} \right)$. Therefore, for any $\gamma$, $ P\left(|\hat{p}_i - p_i| > \gamma\right) \leq  P(|Z| > \Bar{z}(\gamma))$ where $Z \sim \mathcal{N} \left( 0, 1\right)$ and $\Bar{z}(\gamma) = \gamma\sqrt{\frac{n^{\epsilon'}}{p_i(1-p_i)}}$. Let $z'(\gamma) = 2\gamma n^{\frac{\epsilon'}{2}}$. As $p_i(1-p_i) \leq \frac{1}{4}$, $\Bar{z}(\gamma) \geq z'(\gamma)$.    
    Therefore, we can write the following:  
    \begin{equation}\label{eq:gamma}
    \begin{aligned}
    P\left(|\hat{p}_i - p_i| > \gamma \right) 
    &= P(|Z| > \Bar{z}(\gamma)) \leq P(|Z| > z'(\gamma)) \\
    &\leq 2 \int_{z'(\gamma)}^{\infty} \frac{1}{\sqrt{2\pi}} e^{-\frac{t^2}{2}} \, dt \\
    &\leq \sqrt{\frac{2}{\pi}} \int_{z'(\gamma)}^{\infty} e^{-\frac{t}{2}} \, dt \\
    &= 2 \sqrt{\frac{2}{\pi}} e^{-\frac{z'(\gamma)}{2}} \\
    &= 2 \sqrt{\frac{2}{\pi}} e^{-\gamma n^{\frac{\epsilon'}{2}}}.
    \end{aligned}
    \end{equation}

\end{proof}

\begin{proposition}\label{lem:gamma1} For any $\hat{n} = O(n)$,
    $P(\exists i\in[1:k], |p_i - \hat{p_i}| > \frac{\hat{n}^{\epsilon}\sqrt{\hat{n}p_i}}{4\hat{n}}) = o(e^{-n^{\frac{c}{2}}})$ when $p_k =\omega(\frac{n^\epsilon k}{\sqrt{n}})$.
\end{proposition}

\begin{proof}
    For any type $i\in[k]$, let $\gamma_i = \frac{\hat{n}^{\epsilon}\sqrt{\hat{n}p_i}}{4\hat{n}}$. As $\hat{n} = O(n)$, and $\epsilon+\frac{1}{2} < 1$, we have $\gamma_i = \Omega(\frac{n^{\epsilon}\sqrt{np_i}}{4n})$.
    Therefore, 
    $$\gamma n^{\frac{\epsilon'}{2}} = \Omega(n^{\epsilon + \frac{\epsilon'}{2} -\frac{1}{2}}\sqrt{p_{i}}) = \Omega(n^{\epsilon + \frac{\epsilon'}{2} -\frac{1}{2}} \sqrt{p_k}) = \omega(n^{\epsilon + \frac{\epsilon'}{2} -\frac{1}{2}}\sqrt{n^{\epsilon - \frac{1}{2}}k})$$ 
    where we used the bound on $p_k$ for the last equality. As $n^{\epsilon + \frac{\epsilon'}{2} -\frac{1}{2}}\sqrt{n^{\epsilon - \frac{1}{2}}k} \geq n^{1.5\epsilon + \frac{\epsilon'}{2}-\frac{3}{4}}$, by the second inequality in claim \ref{cl:alleqs}, $1.5\epsilon + \frac{\epsilon'}{2}-\frac{3}{4}\geq \frac{c}{2}$, hence $\gamma n^{\frac{n^{\epsilon'}}{2}} = \omega(n^{\frac{c}{2}}).$ According to \cref{lem:probpredict}, for any type $i$, $P(|p_i - \hat{p_i}| > \gamma_i) = o(e^{-n^{\frac{c}{2}}})$. By union bound, $$P(\exists i\in[1:k], |p_i - \hat{p_i}| > \gamma_i) = ko(e^{-n^{\frac{c}{2}}}) = o(e^{-n^{\frac{c}{2}}}).$$ 
\end{proof}

\begin{proposition}\label{lem:gamma2}
    $P(\exists i\in[1:k], |p_i - \hat{p_i}| > \Theta(p_k)) = o(e^{-n^{\frac{2+7c}{18}}})$ when $p_k =\omega(\frac{n^\epsilon k}{\sqrt{n}})$.
\end{proposition}
\begin{proof}
     For $\gamma = \Theta(p_k)$, considering $p_k =\omega(\frac{n^\epsilon k}{\sqrt{n}})$, we obtain $\gamma n^{\frac{\epsilon'}{2}} = \omega(n^{\epsilon + \frac{\epsilon'}{2}-\frac{1}{2}}k)$. As $n^{\epsilon + \frac{\epsilon'}{2}-\frac{1}{2}}k \geq n^{\epsilon + \frac{\epsilon'}{2}-\frac{1}{2}}$, by replacing $\epsilon = \frac{5+4c}{18}$ and $\epsilon' = \frac{2+c}{3}$, we obtain $\gamma n^{\frac{\epsilon'}{2}} = \omega(n^{\frac{2+7c}{18}})$. By applying \cref{lem:probpredict} and union bound, 
     $$P(\exists i\in[1:k], |p_i - \hat{p_i}| > \Theta(p_k)) = o(ke^{-n^{\frac{2+7c}{18}}}) = o(e^{-n^{\frac{2+7c}{18}}}).$$
\end{proof}

\begin{lemma}\label{lem:boundX_iunknown}
For any \(n' = \Omega\left(\frac{n}{k}\right)\) with \(n' \le n\), if \(p_k = \omega\left(\frac{n^\epsilon k}{\sqrt{n}}\right)\), with probability at least \(1 - o(e^{-n^{c/2}})\), for every type \(i \in [k]\) we have 
\[
X^{n'}_i \le n'\hat{p_i} + n'^\epsilon \sqrt{n'\hat{p_i}}.
\]
\end{lemma}
\begin{proof}
    According to \cref{prop:boundx_i}, with probability at least $1-O(e^{-n'^{2\epsilon}})$, for every type $i\in[k]$,
\begin{equation}\label{eq:upperX_i}
    X^{n'}_i \leq n'p_i + n'^{\epsilon}\sqrt{n'p_i}/4.
\end{equation}
As $n' = \Omega(\frac{n}{k})$, and considering the bound on $k$ in \eqref{boundk}, $n' = \omega(n^{\frac{8+c}{9}})$. Hence $n'^{2\epsilon} = \omega(n^{\frac{(8+c)(5+4c)}{81}})$. Therefore the bound on $X^{n'}_i$ given in \eqref{eq:upperX_i} holds with probability at least $1-o(e^{-n^{\frac{(8+c)(5+4c)}{81}}})$.

As $n'= O(n)$, by \cref{lem:gamma1}, with probability at least $1-o(e^{-n^{c/2}})$, $p_i \leq \hat{p_i} + \frac{n'^{\epsilon}\sqrt{n'p_i}}{4n'}$. Inserting this in \eqref{eq:upperX_i},  
\begin{equation}\label{eq:aspecial}
    X^{n'}_i \leq n'\hat{p_i} + n'^{\epsilon}\sqrt{n'p_i}/2.
\end{equation} 
Now using \cref{lem:gamma2}, with probability at least $1-o(e^{-n^{\frac{2+7c}{18}}})$, 
\begin{equation}\label{eq:pminfirst}
    p_i \leq \hat{p_i} + p_k.
\end{equation}

Using \cref{lem:gamma2} again, with probability at least $1-o(e^{-n^{\frac{2+7c}{18}}})$, $p_i - \frac{p_k}{2}\leq\hat{p_i}$. Since $p_k \leq p_{i}$, $\frac{p_k}{2} \leq \hat{p_i}$. Combining this with \eqref{eq:pminfirst}, 
$$p_i \leq 3\hat{p_i}$$
where this holds with probability at least $1-2o(e^{-n^{\frac{2+7c}{18}}})$ by union bound. Inserting this inequality in \eqref{eq:aspecial} we obtain the following: 
\begin{equation}
    \begin{aligned}
        X^{n'}_i & \leq n'\hat{p_i} + n'^{\epsilon}\sqrt{n'p_i}/2 \\
        & \leq  n'\hat{p_i} + \frac{\sqrt{3}}{2}\cdot n'^{\epsilon}\sqrt{n'\hat{p_i}} \\
        & \leq n'\hat{p_i} + n'^{\epsilon}\sqrt{n'\hat{p_i}}.
    \end{aligned}
\end{equation}
Using union bound, for a type $i\in[k]$, $X^{n'}_i \leq n'\hat{p_i} + n'^{\epsilon}\sqrt{n'\hat{p_i}}$ with probability at least $$1-o(e^{-n^{\frac{(8+c)(5+4c)}{81}}}) -o(e^{-n^{c/2}}) - 2o(e^{-n^{\frac{2+7c}{18}}}) = 1-o(e^{-n^{c/2}}).$$ Now applying union bound again over the types, we obtain this upper bound holds for all types with probability at least $1-ko(e^{-n^{c/2}})$. As $k$ is constant or sublinear in $n$, this probability simplifies to $1-o(e^{-n^{c/2}})$. 
\end{proof}

\begin{lemma}\label{lem:B1}
    If $|C| < \lceil n^{\epsilon'} \rceil$, the following bounds hold when \( \delta = \frac{1}{n^{\frac{\epsilon' - c}{2}}} \). 
    \begin{enumerate}
        \item With probability at least $1-O(e^{-n^c})$ for all $i\in[1,k]$, $ v_i(B_1) \geq (1-\delta) 2k\lceil n^{\epsilon'} \rceil$. 
        \item  With probability at least $1-O(e^{-n^{1/3}})$ for all $i\in[1,k]$, $ v_i(B_2) \geq (1-\delta)\cdot (n-|C^0| - 2k\lceil n^{\epsilon'} \rceil)$. 
    \end{enumerate}
\end{lemma}
\begin{proof}
    When $|C| < \lceil n^{\epsilon'} \rceil$, $C^0$ represents the set of initial items in $C$. Before filling \( B_1 \) and \( B_2 \), one unique item of $C^0$ is allocated to each of the first \( |C^0| \) agents. 
    Since each item has a value of at most 1 for every type in a normalized instance, it follows that for all \( i \in [k] \), 
    \[
    v_i(M') \geq n - |C^0|.
    \]
    Each item of \( M' \) is added to \( B_1 \) with probability \( \frac{2k\lceil n^{\epsilon'} \rceil}{n - |C^0|} \), and to \( B_2 \) otherwise. For every type \(i \in [k]\), since the input instance is normalized, each item is valued at most 1 by type \(i\); hence, the valuation function \(v_i\) meets the conditions specified in \cref{cor:nonmonotone-tails}. Therefore, using the concentration bound provided in \cref{cor:nonmonotone-tails} for every bag \( B_j \in \{B_1, B_2\} \), we have:
    \begin{equation}\label{eq:boundB_i}
        P\Bigl(v_i(B_j) < (1 - \delta) \mathbb{E}[v_i(B_j)]\Bigr) \leq e^{-\delta^2\mathbb{E}[v_i(B_j)]/4 }.
    \end{equation}
    Now, note that
    \begin{equation}\label{eq:boundexpB_1}
    \begin{aligned}
        \mathbb{E}[v_i(B_1)] &= p \cdot v_i(M') \\
        &\geq p \cdot (n - |C^0|) = 2k\lceil n^{\epsilon'} \rceil.
    \end{aligned}
    \end{equation}
    Therefore,
    $$
    P\left(v_i(B_1) < (1 - \delta) 2k\lceil n^{\epsilon'} \rceil\right) \leq P\Bigl( v_i(B_1) < (1 - \delta) \mathbb{E}[v_i(B_1)] \Bigr).
    $$
    Using this inequality with \eqref{eq:boundB_i} and \eqref{eq:boundexpB_1}, we obtain
    \[
    P\left(v_i(B_1) < (1 - \delta) 2k\lceil n^{\epsilon'} \rceil\right) \leq e^{-\delta^2 k\lceil n^{\epsilon'} \rceil/2 }.
    \]
    Substituting \( \delta = \frac{1}{n^{\frac{\epsilon' - c}{2}}} \), we obtain:
    \[
    e^{-\delta^2 k\lceil n^{\epsilon'} \rceil/2} = O(e^{-kn^c}).
    \]
    Now, applying the union bound for any $k = O(n)$,
    \[
    P\left(\exists i \in [1,k]: v_i(B_1) < (1 - \delta) 2k\lceil n^{\epsilon'} \rceil\right) \leq O(k e^{-kn^c}) = O(e^{-n^{c}}).
    \]
    Therefore, with probability at least \( 1 - O(e^{-n^c}) \), for all \( i \in [1,k] \), 
    \[
    v_i(B_1) \geq (1 - \delta) 2k\lceil n^{\epsilon'} \rceil.
    \]

    A similar analysis for \( v_i(B_2) \) yields:
    \[
    P\left( v_i(B_2) < (1 - \delta) (n - |C^0| - 2k\lceil n^{\epsilon'} \rceil)\right) \leq e^{-\delta^2(n - |C^0| - 2k\lceil n^{\epsilon'} \rceil)/4}.
    \]
    Since \( |C^0| < \lceil n^{\epsilon'} \rceil \), and considering \( \epsilon' = \frac{2+c}{3} \) along with the upper bound on \( k \) from \eqref{boundk}, we observe that
    \[
    n - |C^0| - 2k\lceil n^{\epsilon'} \rceil = \Theta(n).
    \]
    Substituting \( \delta = \frac{1}{n^{\frac{\epsilon' - c}{2}}} \), we obtain:
    \[
    e^{-\delta^2(n - |C^0| - 2k\lceil n^{\epsilon'} \rceil)/4} = \Theta(e^{-n^{1 - \epsilon' + c}}) = O(e^{-n^{1/3}}).
    \]
    Applying the union bound for any $k = O(n)$, we obtain,
    \begin{align*}  
    P\left(\exists i \in [1,k]: v_i(B_2) < (1 - \delta) (n - |C^0| - 2k\lceil n^{\epsilon'} \rceil)\right) = \\  k O(e^{-n^{1/3}})= O(e^{-n^{1/3}}).
    \end{align*}    
\end{proof}

\subsection{Large number of Universally High-Valued Items (Invoking \cref{alg:knownD})}\label{app:invoke5}
If \(|C| \geq \lceil n^{\epsilon'} \rceil\), then \cref{alg:unknownD} allocates one item from \(C\) to each of the first \(\lceil n^{\epsilon'} \rceil\) arriving agents. Since any singleton bundle containing an item from \(C\) is valued at least \(\alpha\) by all types, the MMS requirement is met for these agents. We then need only to show that every agent in the reduced instance \(\Icalhat\) also receives a bundle valued at least \(\alpha\). Recall that \cref{alg:unknownD} subsequently invokes \cref{alg:knownD} on the reduced instance
\[
\Icalhat = \bigl([\lceil n^{\epsilon'} \rceil + 1, n],\, M' \cup C^0,\, \{v_j\}_{j=1}^k\bigr),
\]
along with the estimated distribution \(\hat{D}\) obtained from observing the types of the first \(\lceil n^{\epsilon'} \rceil\) agents. Let \(j^* \in [k]\) be the type with the highest estimated probability, i.e., \(\hat{p}_{j^*} \ge \hat{p}_j\) for all \(j \in [k]\); hence, \(\hat{p}_{j^*} \ge \frac{1}{k}\). Also, denote by \(\hat{n}\) the number of agents in \(\Icalhat\), so that \(\hat{n} = n - \lceil n^{\epsilon'} \rceil\).

\lemalgfiveworks*
\begin{proof}
    First, note that as $p_k = \omega(\frac{k}{n^{\frac{2}{9}(1-c)}})$, and $\epsilon = \frac{5+4c}{18}$, $p_k =\omega(\frac{n^\epsilon k}{\sqrt{n}})$. 
    
    For a given \(\alpha = \frac{1}{2(1+\eta)}\), \cref{alg:knownD} chooses to invoke either \cref{alg:highC} or \cref{alg:lowC} based on the number of remaining universally high-valued items, \(|C^0|\). Let us analyze each case separately. 

\begin{itemize}
    \item If $|C^0| \geq \hat{n}(1-\frac{1}{k}) + \hat{n}^{\epsilon}\sqrt{\hat{n}\hat{p_{j^*}}}$, \cref{alg:highC} is invoked. By \cref{prop:boundx_i}, we have that with probability at least $1-O(e^{-n^{2\epsilon}})$, 
\begin{equation}\label{eq:jstar}
    X^{\hat{n}}_{j^*} \geq \hat{n}p_{j^*} - \hat{n}^{\epsilon}\sqrt{\hat{n}p_{j^*}}/4.
\end{equation}
Using $\hat{p_{j^*}} \geq \frac{1}{k}$, and \cref{lem:gamma1}, with probability at least $1-o(e^{-n^{c/2}})$, $p_{j^*} \geq \hat{p}_{j^*}  - \frac{\hat{n}^{\epsilon}\sqrt{\hat{n}p_{j^*}}}{4\hat{n}} \geq \frac{1}{k} - \frac{\hat{n}^{\epsilon}\sqrt{\hat{n}p_{j^*}}}{4\hat{n}}$.  Inserting this in \eqref{eq:jstar}, we obtain 
\begin{equation}
    X^{\hat{n}}_{j^*} \geq \frac{\hat{n}}{k} -\hat{n}^{\epsilon}\sqrt{\hat{n}p_{j^*}}/2.
\end{equation}
Now using \cref{lem:gamma2}
, with probability at least $1-o(e^{-n^{\frac{2+7c}{18}}})$, $p_{j^*} \leq \hat{p}_{j^*} + p_k$. Therefore 
\begin{equation}
    X^{\hat{n}}_{j^*} \geq \frac{\hat{n}}{k} -\hat{n}^{\epsilon}\sqrt{\hat{n}p_{j^*}}/2 
    \geq \frac{\hat{n}}{k} -\hat{n}^{\epsilon}\sqrt{\hat{n}(\hat{p}_{j^*} + p_k})/2.
\end{equation}
Note that since $j^*$ is the most frequent type, $\hat{p_{j^*}} \geq \frac{1}{k}$, while $p_k$ is the probability of the least frequent type, which is always less than $\frac{1}{k}$. Therefore, $p_k \leq \hat{p}_{j^*}$ and 
\begin{equation}
\begin{aligned}    
    X^{\hat{n}}_{j^*}  &
    \geq \frac{\hat{n}}{k} -\frac{\sqrt{2}}{2} \cdot \hat{n}^{\epsilon}\sqrt{\hat{n}\hat{p}_{j^*}} \\ &\geq \frac{\hat{n}}{k} -\hat{n}^{\epsilon}\sqrt{\hat{n}\hat{p}_{j^*}} \\
    & = \hat{n} - \left[\hat{n}(1-\frac{1}{k}) +  \hat{n}^{\epsilon}\sqrt{\hat{n}\hat{p}_{j^*}}\right] \\ 
    & \geq \hat{n} - |C^0|.
    \end{aligned}
\end{equation}
By union bound, this bound on $X^{\hat{n}}_{j^*}$ holds with probability at least $1 - O(e^{-n^{2\epsilon}}) -  o(e^{-n^{c/2}}) - o(e^{-n^{\frac{2+7c}{18}}}) = 1-o(e^{-n^{c/2}})$. Given this holds, \cref{lem:highC}, implies that \cref{alg:highC} is $\alpha-$MMS competitive. 

\item If $|C^0| < \hat{n}(1-\frac{1}{k}) + \hat{n}^{\epsilon}\sqrt{\hat{n}\hat{p_{j^*}}}$, each item of $C^0$ will be allocated to a unique agent from the first $|C^0|$ agents that arrive in $\Icalhat$, and then the problem is reduced to a smaller instance $\I'$ with $n' = n -\lceil n^{\epsilon'} \rceil - |C^0|$ agents, and items in $M'$. Since $c<1$, $\epsilon' = \frac{2+c}{3} < 1$; thereby $n - \lceil n^{\epsilon'} \rceil = \Theta(n)$, implying $n' = \Theta(n-|C^0|)$. 

As $\hat{n} \leq n$, and $\hat{p}_{j^*} \leq 1$, $|C^0| < n(1-\frac{1}{k}) + n^{\epsilon}\sqrt{n}$. Given this bound on $|C^0|$, $n' = \Omega(\frac{n}{k} - n^{\epsilon}\sqrt{n})$. Given the bound on $k$ in \cref{boundk}, $\frac{n}{k} = \Omega(n^{\frac{8+c}{9}})$, and given $\epsilon= \frac{5+4c}{18}$, $n^{\epsilon}\sqrt{n} = \Theta(n^{\frac{7+2c}{9}})$. As $c<1$, $n' = \Omega(\frac{n}{k})$. Let \(X^{n'}_i\) denote the ex-post number of type \(i\) agents in \(\I'\). According to \cref{lem:boundX_iunknown}, with probability at least $1-o(e^{-n^{c/2}})$, for every type $i\in[k]$, 
\begin{equation}\label{eq:boundX_in'}
    X^{n'}_i \leq n'\hat{p_i} + n'^{\epsilon}\sqrt{n'\hat{p}_i}.
\end{equation} 
As \cref{alg:lowC} is invoked on $\I'$ with $n'$ agents, and the estimated distribution $\hat{D}$, the upper limit on the number of reserved bundles for each type $i\in[k]$ is fixed as $M_i= \lfloor n'\hat{p}_i + n'^{\epsilon}\sqrt{n'\hat{p}_i}\rfloor$. Hence, \cref{eq:boundX_in'} implies with probability at least $1-o(e^{-n^{c/2}})$, for every $i\in[k]$, $ X^{n'}_i \leq M_i$. 

Note that before invoking \cref{alg:lowC}, we have only allocated $\lceil n^{\epsilon'} \rceil + |C^0|$ items, where the value of each item is bounded by $1$ in a normalized input, therefore $v_i(M') \geq n- \lceil n^{\epsilon'} \rceil - |C^0|= n'$ for all $i\in[k]$. Hence, according to \cref{lem:allsaturated} there exists an \( n(\eta) \) such that for all \( n > n(\eta) \), every type becomes saturated during the preprocessing phase—that is, \( |G_i| = M_i \) for all \( i \in [k] \). This guarantees that, with high probability, the number of reserved bundles for each type is at least equal to the number of agents of that type observed ex-post.
\end{itemize}

\end{proof}

\subsection{Learning with Small Number of Universally High-Valued Items (Invoking \cref{alg:adversarial-alg})}\label{app:invoke2}
\lemalgtwoworks*
\begin{proof}
Note that \(\Icalhat = ([|C^0|+1,\, |C^0| + \lceil n^{\epsilon'} \rceil],\, B_1,\, \{v_j\}_{j=1}^k)\) is the instance on which \cref{alg:adversarial-alg} is invoked, and let \(\hat{n} = \lceil n^{\epsilon'} \rceil\) denote the number of agents in \(\Icalhat\). By \cref{lem:B1}, with probability at least \(1 - O(e^{-n^c})\), for all \(i\in[k]\) we have \(v_i(B_1) \ge (1-\delta)2k\hat{n}\). Given this, we now show that the MMS value for every type in \(\Icalhat\) is at least \(k(1-\delta)\), i.e., for all \(i\in[k]\), \(\MMS_{v_i}(\Icalhat) \ge k(1-\delta)\). To do so, consider any type \(i\in[k]\) and perform an arbitrary bag-filling procedure on \(B_1\) that stops as soon as a bag's value for type \(i\) reaches \(k(1-\delta)\); note that each such bag has value at most \(k(1-\delta)+1\) (since each item is valued at most \(1\) in a normalized instance). Thus, the number of bags that can be formed is at least 
\[
\left\lfloor \frac{v_i(B_1)}{k(1-\delta)+1} \right\rfloor.
\]
Given $\delta$ as \eqref{dfn:delta} and $\epsilon' = \frac{2+c}{3}$, for any $c < 0.1$, and $n>10$, we have that $\delta = \frac{1}{n^{\frac{1-c}{3}}} \leq \frac{1}{11^{0.3}} \leq 0.5$. Given $k\geq 2$, \(k(1-\delta)+1 \leq 2k(1-\delta)\); together with \(v_i(B_1) \ge (1-\delta)2k\hat{n}\) this implies 
\[
\left\lfloor \frac{v_i(B_1)}{k(1-\delta)+1} \right\rfloor \ge \hat{n}.
\]
Hence, for each type \(i\) we can construct at least \(\hat{n}\) bags, each with value at least \(k(1-\delta)\), which in turn implies \(\MMS_{v_i}(\Icalhat) \ge k(1-\delta)\). Finally, by \cref{thm:adv}, \cref{alg:adversarial-alg} is $1/k-$MMS competitive, so every agent in $\Icalhat$ receives a bundle with value at least $\frac{1}{k}$ times their MMS, which is at least $1-\delta$. Since for \(n>10\) and $c<0.1$, we have \(\delta \leq 0.5\) (so that \(1-\delta \ge \frac{1}{2} \ge \alpha\)), it follows that every agent in \(\Icalhat\) receives a bundle valued at least \(\alpha\) with probability at least \(1 - O(e^{-n^c})\).
\end{proof}

\subsection{Allocation with Small Number of Universally High-Valued Items (Invoking \cref{alg:lowC})}\label{app:invokeend}
    Note that \cref{alg:lowC} is invoked only when the initial number of items in $C$, is less than $\lceil n^{\epsilon'} \rceil$, i.e. $|C^0| < \lceil n^{\epsilon'} \rceil$. In this case, \cref{alg:lowC} is invoked on $\Ical' = ([|C^0|+\lceil n^{\epsilon'} \rceil + 1,n], B_2, \{v_j\}^k_{j=1})$ with the estimated distribution $\hat{D}$. We adhere to the definitions of \(n'\), \(\{\mu_j\}_{j=1}^k\), \(z\), \(b\), \(E_j\), and \(\Bar{T}_i\) as given in \cref{sub:knownDnotations}. In this context, with \(\I'\) as the input instance and \(\hat{D}\) as the known distribution, \(n'\) denotes the number of agents in \(\I'\); specifically, 
\[
n' = n - |C^0| - \lceil n^{\epsilon'} \rceil.
\]
Moreover, the expected number of agents of type $i$ is estimated as $\mu_i = n'\hat{p}_i$. The upper-limit on any $|G_i|$, number of bundles reserved for type $i$ in the preprocessing steps of \cref{alg:lowC}, is $M_i = \lfloor\mu_i + n'^{\epsilon}\sqrt{\mu_i}\rfloor$.  Note that after the preprocessing, a type $i$ is saturated if $|G_i| = M_i$. Finally, by \eqref{definez}, $z = n'^\epsilon \sum_{j=1}^{k} \sqrt{\mu_j}$.

\begin{proposition}\label{prop:knz}
Assuming \(n' = \Theta(n)\) and \(p_k = \omega\Bigl(\frac{n^\epsilon k}{\sqrt{n}}\Bigr)\), with probability at least \(1 - o\Bigl(e^{-n^{\frac{2+7c}{18}}}\Bigr)\) the following asymptotic relations hold:
 \begin{enumerate}
     \item  $k\lceil n^{\epsilon'} \rceil = o(z)$.
    \item $\delta n'= o(z)$.
 \end{enumerate}
\end{proposition}
\begin{proof}
    Using \cref{lem:gamma2}, with probability at least $1-o(e^{-n^{\frac{2+7c}{18}}})$, for any type $i\in[k]$, $p_i - \frac{p_k}{2}\leq\hat{p_i}$. Since $p_k \leq p_{i}$, $\frac{p_k}{2} \leq \hat{p_i}$. Using this in \eqref{definez}, 
    \begin{equation}
        n'^{\epsilon}k\sqrt{n'p_k/2}\leq z. 
    \end{equation}

    Since $p_k = \omega(\frac{n^{\epsilon}k}{\sqrt{n}})$, and $n'=\Theta(n)$,
    \begin{equation}\label{eq:lowerboundz}
        z = \omega(n^{\epsilon}k\sqrt{n^{\epsilon+1/2}k}) = \omega(kn^{1.5\epsilon + 1/4}).
    \end{equation}
    Hence,     
    $$\frac{k\lceil n^{\epsilon'} \rceil}{z} = o( \frac{kn^{\epsilon'}}{kn^{1.5\epsilon + 1/4}}) = o(n^{\epsilon' - 1.5\epsilon-1/4}).$$ Therefore, we need $\epsilon' \leq 1.5\epsilon + \frac{1}{4}$, where by \cref{cl:alleqs}, this holds, so $\frac{k\lceil n^{\epsilon'} \rceil}{z} = o(1)$. 

    Based on \eqref{eq:lowerboundz}, \eqref{dfn:delta}, and $n' = \Theta(n)$, 
    $$
    \frac{\delta n'}{z} = o(\frac{n^{1-\epsilon'/2 + c/2}}{n^{1.5\epsilon + 1/4}}) = o(n^{3/4 -\epsilon'/2 - 1.5\epsilon + c/2}).
    $$
    By \cref{cl:alleqs}, $3/4 -\epsilon'/2 - 1.5\epsilon + c/2 \leq 0$, hence $\frac{\delta n'}{z} = o(1)$.  
\end{proof}

\lemalgendworks*
\begin{proof}
    
    As $|C^0| < \lceil n^{\epsilon'} \rceil$, $n'\geq n - 2\lceil n^{\epsilon'} \rceil$. Given $\epsilon' < 1$, $n' = \Theta(n)$. Hence, \cref{lem:boundX_iunknown} implies with probability at least $1-o(e^{-n^{c/2}})$, the number of agents arriving from type $i$ in $\I'$ is at most $M_i$, i.e. $X^{n'}_i \leq M_i$. Therefore, it suffices to show that by the end of preprocessing, all types are saturated, ensuring that each type has reserved at least as many bundles as there are agents of that type.

    Suppose for contradiction that a type $i$ is not saturated by the end of preprocessing step. According to \cref{lem:oneunsat}, 
    \begin{equation}
        |G_i| \geq |T_i| - n' + \frac{ v_i(B_2)-|T_i|}{2\alpha} + n'\hat{p_i} - n'^\epsilon \sum_{j\neq i}\sqrt{n'\hat{p_j}}.
    \end{equation}
    Given in this problem $\mu_i = n'\hat{p}_i$, 
    \begin{equation}\label{eq:boundG_i}
        |G_i| \geq |T_i| - n' + \frac{ v_i(B_2)-|T_i|}{2\alpha} + n'\hat{p_i} - n'^\epsilon \sum_{j\neq i}\sqrt{\mu_j}.
    \end{equation}

    According to \cref{lem:B1} with probability at least $1-O(e^{-n^{1/3}})$, for every type $i\in[k]$, $v_i(B_2)\geq (1-\delta)\cdot (n-|C^0| - 2k\lceil n^{\epsilon'} \rceil)$. Given $n' =  n - |C^0| - \lceil n^{\epsilon'} \rceil$, $v_i(B_2) \geq (1-\delta)\cdot(n' -2k\lceil n^{\epsilon'} \rceil)$. Using this in \cref{eq:boundG_i}, we obtain
    
     \begin{equation}\label{eq:boundG_i1}
     \begin{aligned}
        |G_i| & \geq |T_i| - n' + \frac{ (1-\delta)\cdot(n' -2k\lceil n^{\epsilon'} \rceil)-|T_i|}{2\alpha} + n'\hat{p_i} - n'^\epsilon \sum_{j\neq i}\sqrt{\mu_j}\\
        &  = |T_i| - n' + \frac{(1-\delta)(n'-|T_i|)}{2\alpha} - \frac{(1-\delta)2k\lceil n^{\epsilon'} \rceil}{2\alpha} \\& -\frac{\delta |T_i|}{2\alpha} +  n'\hat{p_i} - n'^\epsilon \sum_{j\neq i}\sqrt{\mu_j} \\
  & = (n' - |T_i|) \left[\frac{1-\delta}{2\alpha}-1 \right] -\frac{(1-\delta)2k\lceil n^{\epsilon'} \rceil}{2\alpha} - \frac{\delta |T_i|}{2\alpha} \\& + n'\hat{p_i} - n'^\epsilon \sum_{j\neq i}\sqrt{\mu_j}.
    \end{aligned}
    \end{equation}

Since our chosen $\alpha$ will be arbitrarily close to $\frac{1}{2}$, $\alpha \geq \frac{1}{4}$. Therefore, 
\begin{equation}\label{eq:readyforT_i}
\begin{aligned}
    |G_i| &  \geq (n' - |T_i|) \left[\frac{1-\delta}{2\alpha}-1 \right] -4k(1-\delta)\lceil n^{\epsilon'} \rceil- 2\delta |T_i| \\ &+ n'\hat{p_i} - n'^\epsilon \sum_{j\neq i}\sqrt{\mu_j}
\end{aligned}
\end{equation}
Let \(j^* \in [k]\) be the type with the lowest estimated probability (i.e., \(\hat{p}_{j^*} \le \hat{p}_j\) for all \(j \in [k]\)); since type \(i\) is not saturated, \cref{prop:limitT_i} implies that \(|T_i| \le n'\Bigl(1 - \frac{\hat{p}_{j^*}}{2}\Bigr) + z\), where \(z = n'^\epsilon \sum_{j=1}^{k} \sqrt{\mu_j}\). Moreover, by \cref{lem:gamma2}, with probability at least \(1 - o\Bigl(e^{-n^{\frac{2+7c}{18}}}\Bigr)\) we have \(p_{j^*} - \frac{p_k}{2} \le \hat{p}_{j^*}\); since \(p_k \le p_{j^*}\), it follows that \(\frac{p_k}{2} \le \hat{p}_{j^*}\). Therefore, 
\[
|T_i| \le n'\Bigl(1 - \frac{p_k}{4}\Bigr) + z.
\]
Using this in \cref{eq:readyforT_i} we obtain
\begin{equation}
\begin{aligned}
    |G_i| &  \geq  \left[\frac{n'p_k}{4} - z \right] \left[\frac{1-\delta}{2\alpha}-1 \right] -4k(1-\delta)\lceil n^{\epsilon'} \rceil\\ &- 2\delta \left[n'-\frac{n'p_k}{4} + z \right] + n'\hat{p_i} -n'^\epsilon \sum_{j\neq i}\sqrt{\mu_j}
\end{aligned}
\end{equation}
As \(p_k = \omega\Bigl(\frac{k}{n^{\frac{2}{9}(1-c)}}\Bigr)\) and \(\epsilon = \frac{5+4c}{18}\), it follows that \(p_k = \omega\Bigl(\frac{n^\epsilon k}{\sqrt{n}}\Bigr)\). Since \(n' = \Theta(n)\), we deduce that \(p_k = \omega\Bigl(\frac{n'^\epsilon k}{\sqrt{n'}}\Bigr)\). By the Cauchy-Schwarz inequality, we have \(z \le n'^\epsilon \sqrt{k}\sqrt{\sum_{i=1}^{k}\mu_i} = n'^\epsilon\sqrt{k n'}\), and hence \(n' p_k = \omega(z)\). Moreover, using \cref{prop:knz} with probability at least \(1-o\Bigl(e^{-n^{\frac{2+7c}{18}}}\Bigr)\), we obtain
\[
4k(1-\delta)\lceil n^{\epsilon'} \rceil = o(z).
\] 

Moreover, since \(n' p_k = \omega(z)\), we have \(\frac{n'p_k}{4} - z = \omega(z) - z\), which is positive for sufficiently large \(n\). Hence,
\[
2\delta\left[n' - \frac{n'p_k}{4} + z\right] = O(2\delta n').
\]
Furthermore, by \cref{prop:knz}, with probability at least \(1 - o\Bigl(e^{-n^{\frac{2+7c}{18}}}\Bigr)\) this term is also \(o(z)\). Consider $\alpha = \frac{1}{2(1+\eta)}$, 
\begin{equation}
    \begin{aligned}
    |G_i| &  \geq  \eta\left[\omega(z) - z \right] -o(z)- o(z) + n'\hat{p_i} -n'^\epsilon \sum_{j\neq i}\sqrt{\mu_j}.
\end{aligned}
\end{equation}
There exists an $n(\eta)$ where for any $n>n(\eta)$, $\eta\left[\omega(z) - z \right] -o(z)- o(z) \geq z$, then,
\begin{equation}
    \begin{aligned}
    |G_i| &  \geq  z + n'\hat{p_i} -n'^\epsilon \sum_{j\neq i}\sqrt{\mu_j}= n'\hat{p_i} + n'^{\epsilon}\sqrt{n'\hat{p_i}} \geq M_i.
\end{aligned}
\end{equation}
 
Using the union bound, with probability at least 
\[
1 - O\Bigl(e^{-n^{1/3}}\Bigr) - 3o\Bigl(e^{-n^{\frac{2+7c}{18}}}\Bigr),
\]
we have \(|G_i| \ge M_i\) for type \(i\). This contradicts the assumption that type \(i\) is unsaturated (i.e., \(|G_i| < M_i\)); hence, all types must be saturated. Moreover, with probability at least \(1-o\Bigl(e^{-n^{c/2}}\Bigr)\) we have \(X^{n'}_i \le M_i\) for every \(i\in[k]\). Therefore, by applying the union bound, with probability at least
\[
1 - o\Bigl(e^{-n^{c/2}}\Bigr) - O\Bigl(e^{-n^{1/3}}\Bigr) - 3o\Bigl(e^{-n^{\frac{2+7c}{18}}}\Bigr) = 1 - o\Bigl(e^{-n^{c/2}}\Bigr),
\]
\cref{alg:lowC} is \(\alpha\)-MMS competitive.

\end{proof}

\section{Discussion}
In this paper we initiate the study of online discrete fair division with agent arrival. It is well-known that this model is not amenable to positive results \cite{amanatidis2023fair}. To circumvent this barrier, we introduce the information structure of \textsc{OnlineKTypeFD}, where every arriving agent belongs to one of $k$ known types and has the corresponding valuation function. Leveraging this, we design constant-MMS-competitive algorithms under the stochastic arrival of agents. 

It would be interesting to study other fairness notions under this model, such as EF1, EFX, and Prop1, and see if a constant approximation is possible for either of these. Another interesting question is, what if we are given a distribution over valuation functions for each type, instead of the exact function? Are reasonable ex-post fairness guarantees still possible with high probability, or we need to consider ex-ante guarantees?
~\newpage

\bibliographystyle{alpha}
\bibliography{main}

\end{document}